%% file: main.tex
\documentclass{biometrika}

\usepackage{amsmath}

\usepackage{newtxtext}
\usepackage{newtxmath}

\usepackage[hidelinks]{hyperref}
\usepackage[capitalise]{cleveref}

\usepackage{tikz}
\usetikzlibrary{arrows.meta,patterns,decorations.pathreplacing}

\usepackage[plain,noend, linesnumbered]{algorithm2e}

\newcommand{\dee}{\mathrm{d}}

\def\bone{\mathbf{1}}
\def\Ind{\bone}

\makeatletter
\def\lstAZ{A, B, C, D, E, F, G, H, I, J, K, L, M, N, O, P, Q, R, S, T, U, V, W, X, Y, Z}
\def\lstaz{a, b, c, d, e, f, g, h, i, j, k, l, m, n, o, p, q, r, s, t, u, v, w, x, y, z}
\def\lstAZBB{B, C, D, E, F, G, H, I, J, K, L, M, N, O, Q, R, T, U, V, W, X, Y, Z}
\newcommand{\MkScr}[1]{\expandafter\def\csname s#1\endcsname{\mathscr{#1}}}
\newcommand{\MkUp}[1]{\expandafter\def\csname u#1\endcsname{\mathrm{#1}}}
\newcommand{\MkFrak}[1]{\expandafter\def\csname f#1\endcsname{\mathfrak{#1}}}
\newcommand{\MkCal}[1]{\expandafter\def\csname c#1\endcsname{\mathcal{#1}}}
\newcommand{\MkBB}[1]{\expandafter\def\csname #1#1\endcsname{\mathbb{#1}}}
\newcommand{\Reals}{\mathbb{R}}

\DeclareMathOperator{\AC}{AC}
\DeclareMathOperator{\HS}{HS}
\DeclareMathOperator{\DS}{DS}

\def\Ind{\bone}

\@for\i:=\lstAZ\do{%
	\expandafter\MkScr \i  %
	\expandafter\MkFrak \i  %
	\expandafter\MkUp \i %
	\expandafter\MkCal \i  %
		  }    
\@for\i:=\lstaz\do{%
	\expandafter\MkUp \i   }    
	
\@for\i:=\lstAZBB\do{%
	\expandafter\MkBB \i     }
\makeatother
\newcommand{\PP}{\mathbb{P}}

\newcommand{\Var}[2][]{\operatorname{Var}_{#1}\left[#2\right]}

\makeatletter
\def\ps@plain{%
  \firstpage{\thepage}%
  \let\@mkboth\@gobbletwo
  \let\@oddhead\@empty
  \let\@evenhead\@empty
  \let\@oddfoot\@empty
  \let\@evenfoot\@empty
}
\let\printhistory\@empty
\makeatother

\begin{document}

\markboth{S. Grazzi et al.}{Sub-Cauchy Sampling}

\title{Sub-Cauchy Sampling: \\Escaping the Dark Side of the Moon}

\author{Sebastiano Grazzi}
\affil{Department of Decision Sciences and BIDSA, Bocconi University
\email{sebastiano.grazzi@unibocconi.it}}

\author{Sifan Liu}
\affil{Department of Statistical Science, Duke University
\email{sifan.liu@duke.edu}}

\author{Gareth O.~Roberts}
\affil{Department of Statistics, University of Warwick
\email{gareth.o.roberts@warwick.ac.uk}}

\author{Jun Yang}
\affil{Department of Mathematical Sciences, University of Copenhagen
\email{jy@math.ku.dk}}

\maketitle

\begin{abstract}
We introduce a Markov chain Monte Carlo algorithm based on Sub-Cauchy Projection, a geometric transformation that generalizes stereographic projection by mapping Euclidean space into a spherical cap of a hyper-sphere, referred to as the complement of the dark side of the moon\footnote{term chosen in homage to Pink Floyd’s \emph{The Dark Side of the Moon} (1973).}. We prove that our proposed method is uniformly ergodic for sub-Cauchy targets, namely targets whose tails are at most as heavy as a multidimensional Cauchy distribution, and show empirically its performance for challenging high-dimensional problems. The simplicity and broad applicability of our approach open new opportunities for Bayesian modeling and computation with heavy-tailed distributions in settings where most existing methods are unreliable.
\end{abstract}


\section{Introduction}

\subsection{Sampling heavy-tailed distributions}

Heavy-tailed distributions are characterized by tails that decay polynomially rather than exponentially, thereby assigning much larger probabilities to extreme values compared to light-tailed distributions such as the exponential or Gaussian. This class of distributions arises naturally in many fields, including finance, insurance, and climate science, where extreme events occur more frequently and have large impact. For such models, the likelihood is typically heavy-tailed as a function of both the data and the model parameters, which frequently induces a heavy-tailed posterior distribution. Heavy-tailed posteriors also arise in binary regression problems under separation \citep{gelman2008weakly, ghosh2018use} and in robust estimation.  As originally noted by \citet{de1961bayesian}, for heavy-tailed models, the contribution of extreme observations to the first two moments of the posterior distribution (when they exist) becomes negligible as the observation magnitude increases. This observation was later formalized in a general setting within the framework of \emph{Bayesian conflict resolution}, which studies the rather subtle and delicate interplay between prior and likelihood specification in the presence of outliers.
See \citet{o2012bayesian} for a comprehensive review and, e.g., \citet{gagnon_theoretical_2023, de2025robust, agrawal2026micromodes} and references therein for successful applications of robust Bayesian estimation with heavy-tailed distributions.

Despite these modeling advances, many challenges remain on the computational side. For heavy-tailed posteriors, most popular  Metropolis--Hastings algorithms are not geometrically ergodic \citep{jarner2000geometric, livingstone2019geometric, wang2025stereographic}. {More recently, \citet{brevsar2025central} establish conditions under which the associated ergodic averages fail to satisfy a central limit theorem.}

For specific statistical models (particularly regression problems), the most popular computational approach is based on Gibbs sampling with an appropriate data augmentation scheme \citep{polson2013bayesian, ghosh2018use, liu2004robit, gagnon_theoretical_2023}. These algorithms can be applied only when simulating from the full conditional distributions is possible, thereby restricting the set of models and priors. Furthermore, convergence guarantees of the Gibbs sampler for heavy-tailed targets are limited and empirical studies suggest slow mixing \citep{ghosh2018use}.

As a consequence, the lack of broadly applicable and reliable sampling methods has limited the use of heavy-tailed models in Bayesian practice, despite their well-recognized statistical advantages. This paper directly addresses this gap by providing a simple, principled, and generally applicable approach for sampling from heavy-tailed posterior distributions. 
\subsection{Overview of our contributions}

At a high level, the proposed Markov chain Monte Carlo (MCMC) algorithm proceeds by first mapping the original state space $\RR^d$ onto a spherical cap of a unit sphere, thereby transforming the target distribution on $\RR^d$ into a distribution supported on a subset of the sphere. A new random-walk-type Metropolis algorithm is then used to sample the target on the spherical cap, and the resulting samples are mapped back to the original space $\RR^d$ via the inverse transformation.

The inverse mapping from the spherical cap to $\RR^d$ is termed the \emph{Sub-Cauchy Projection} (SCP). To describe it precisely, we place a $d$-dimensional unit sphere, which lives in $\RR^{d+1}$, \emph{above} the $d$-dimensional hyperplane 
$
\{ (y, 0) : y \in \RR^d \} \subset \RR^{d+1}.
$
The hyperplane is tangent to the sphere at the origin \( (0, \dots, 0)  \in \RR^{d+1} \), which corresponds to the south pole of the sphere. An \emph{observer} \( o  \in \RR^{d+1}\) is placed \emph{inside} this sphere, and the portion of the spherical surface lying above the observer is removed. The removed spherical cap is hereafter referred to as the \emph{dark side of the moon}, and its complement is called the \emph{bright side}. 
For a point $x$ on the bright side of the sphere, the Sub-Cauchy Projection \( y = \mathrm{SCP}_o(x) \in \RR^d \) is defined as the intersection of the line passing through $o$ and $x$ with the hyperplane $\{(y,0):y\in\RR^d\}$.
See \cref{fig: illustration projection} for an illustration with two different choices of the observer $o$.
 
For every observer \( o \), the map \( \mathrm{SCP}_o\) is a smooth and bijective map from the bright side to $\RR^d$ (a diffeomorphism), although it does not preserve distances and volumes. When \( o = (0, \dots, 0, 2) \), corresponding to the north pole of the sphere, we recover the classical stereographic projection\footnote{Stereographic projection can also be defined with the sphere centered at the origin (see \cite[Fig.~2]{yang2022stereographic}). These two definitions are equivalent up to a rescaling constant.} \citep[see e.g.][]{coxeter1961introduction}, and the dark side degenerates to the singleton \( o \). For any other choice of \( o \) inside the sphere, however, the Sub-Cauchy Projection differs from the classical stereographic projection in two important ways:
\begin{itemize}
    \item The dark side is a $d$-dimensional region, rather than a point, and points at infinity in $\RR^d$ are mapped to its $(d-1)$-dimensional boundary. This difference is reflected in the Jacobian of the transformation (\cref{rmk: jacobian}) and enables SCP to accommodate target distributions with heavier tails than those that can be handled by classical stereographic projection (\cref{rmk: limitation SPS}).
  
    \item The observer $o$ acts as a tuning parameter that provides additional flexibility beyond classical stereographic projection. In particular, by varying the location of $o$, the map $\mathrm{SCP}_o$ can be skewed toward certain direction, making it well suited for sampling from asymmetric or skewed heavy-tailed distributions.
  
\end{itemize}

Our proposed algorithm, termed the \emph{Sub-Cauchy Projection Sampler} (SCS), is a random-walk-type Metropolis algorithm on the bright side. 
Our main theoretical result establishes that SCS is \emph{uniformly ergodic} for any $d$-dimensional \emph{sub-Cauchy distributions}, that is, distributions with tails which are at most as heavy as a $d$-dimensional Cauchy distribution, see \cref{def: sub-cauchy}. To the best of our knowledge, SCS is the first general-purpose Metropolis-adjusted algorithm with provable uniform ergodicity guarantees for multidimensional sub-Cauchy distributions. 
\begin{definition}\label{def: sub-cauchy}
A distribution $\pi$ in $\RR^d$ is sub-Cauchy if 
    \[
    \sup_{y\in \mathbb{R}^d}\pi(y)(1+\|y\|^2)^{(d+1)/2}<\infty.
    \]    
\end{definition}

Beyond the theoretical guarantees, we provide practical guidelines for selecting the observer parameter $o$, along with additional tuning parameters such as shifting and rescaling constants. We further demonstrate the performance of SCS on a range of high-dimensional heavy-tailed target distributions, including examples arising from robust Bayesian binary regression, where data-augmentation Gibbs samplers are known to mix slowly \citep{ghosh2018use}.

{ For heavy-tailed target distributions, standard algorithms that are not uniformly ergodic may explore the bulk of the distribution reasonably well while mixing poorly in the tails. Consider, for example, the Metropolis-adjusted Langevin algorithm (MALA) targeting a Cauchy distribution. Although MALA leaves the target invariant exactly, it is not geometrically ergodic in this setting, and its dynamics approach those of a pure random walk in the tails, leading to poor mixing \citep{roberts1996exponential}. In contrast, SCS is empirically observed to explore heavy-tailed distributions more effectively irrespective of the initialization, yielding more accurate Monte Carlo estimates, particularly for test functions that are sensitive to the tails; see, for example, \cref{sec:skew-cauchy}.}

\begin{figure}
    \centering
\includegraphics[width=0.4\linewidth]{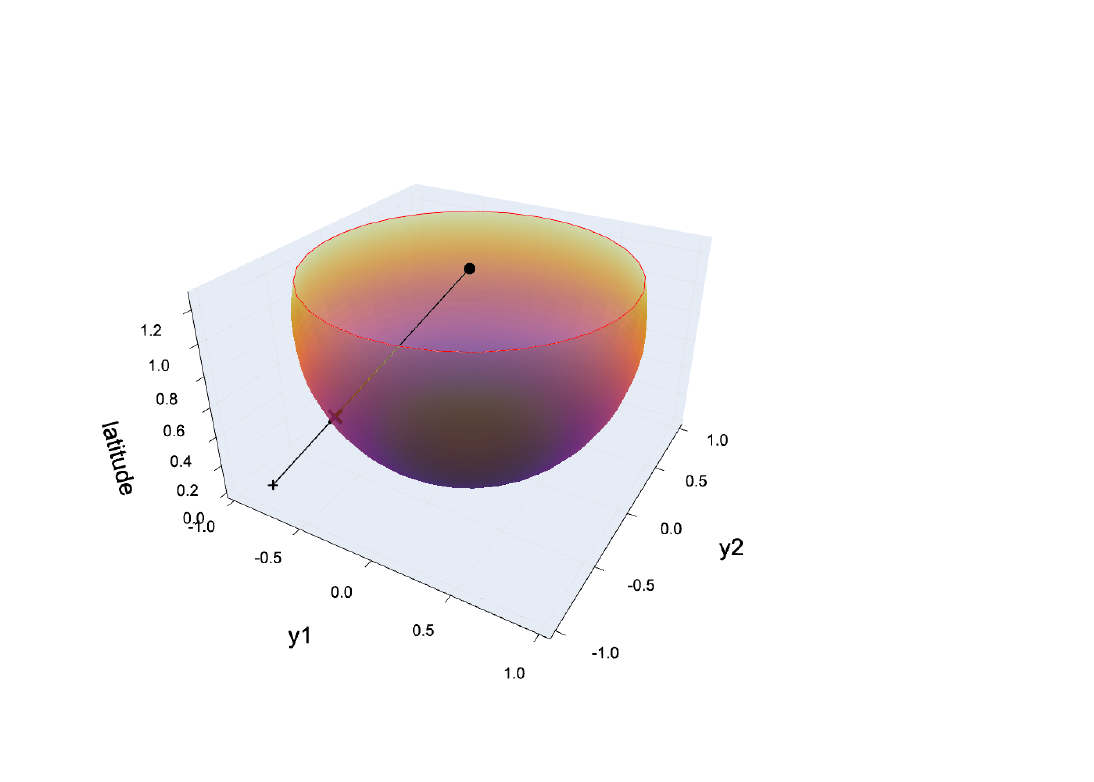}
\includegraphics[width=0.4\linewidth]{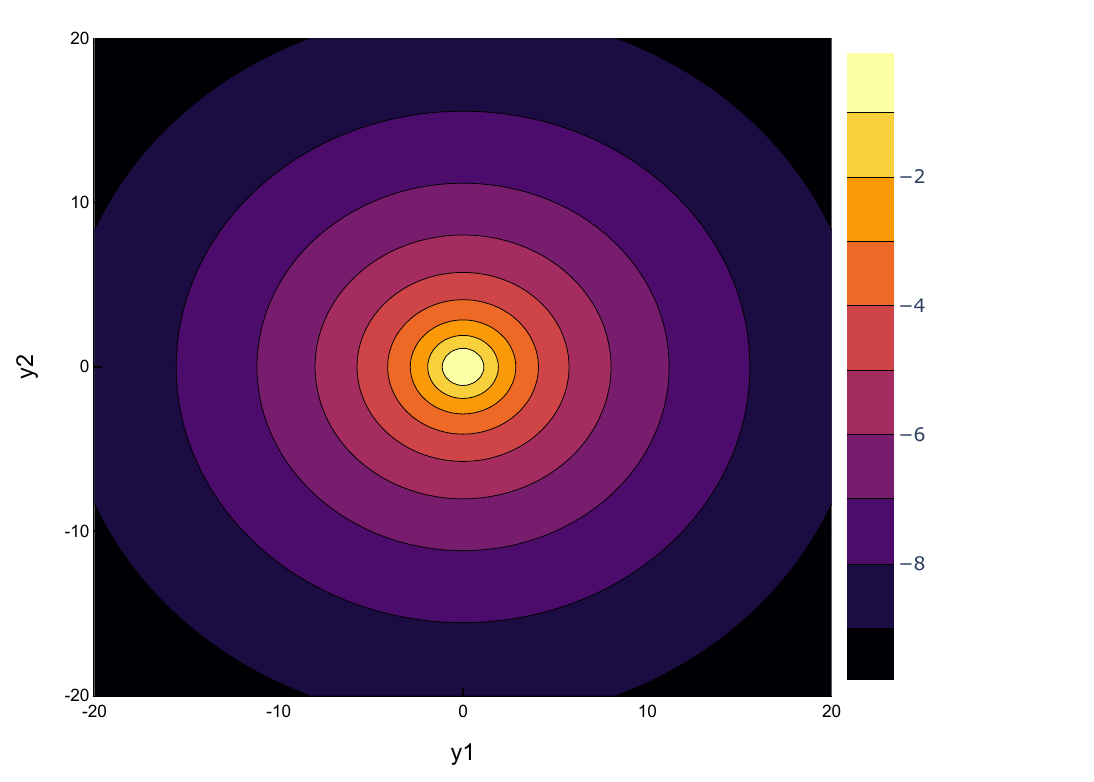}
\includegraphics[width=0.4\linewidth]{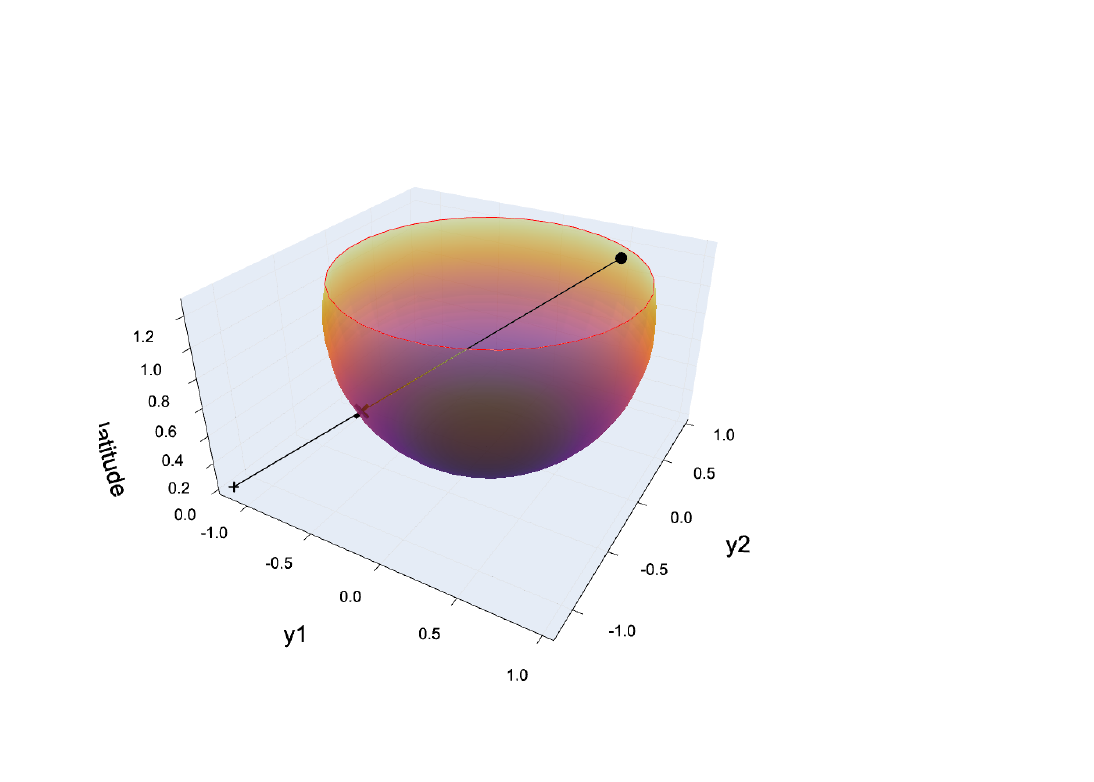}
\includegraphics[width=0.4\linewidth]{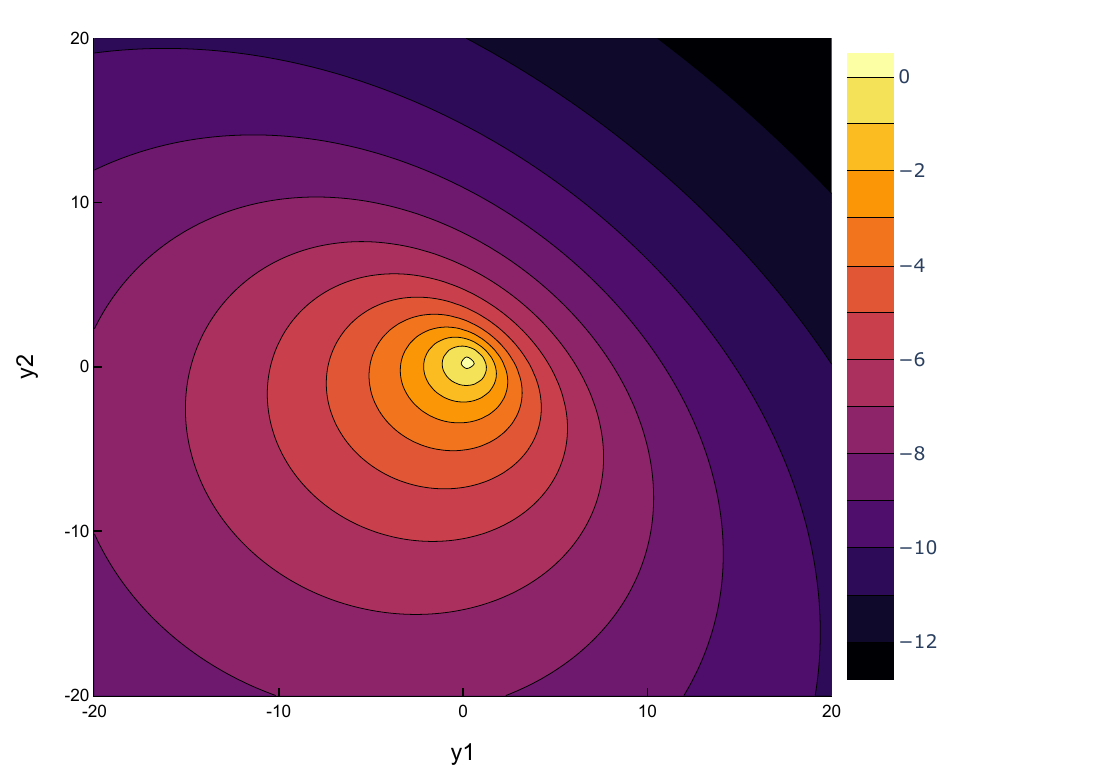}
    \caption{Illustration of a 2-dimensional Sub-Cauchy Projection with $o = (0, 0, 1.5)$ (top panels) and $o = (0.6, 0.6, 1.5)$ (bottom panels). Right panels: contours of log density of the uniform distribution on the bright side of the sphere projected onto $\RR^2$.}
    \label{fig: illustration projection}
\end{figure}

\subsection{Illustration on the Cauchy distribution}\label{sec: dark side of the moon}

Before going into the details of our method, we present an illustrative example to demonstrate the advantages of SCS and to explain the underlying intuition.
We compare the performance of  SCS with other popular MCMC methods for a prototypical example of a heavy-tailed target distribution: the $100$-dimensional Cauchy distribution. For this target, the Stereographic Projection Sampler (SPS; \citealp[Algorithm 1]{yang2022stereographic}) and standard MCMC algorithms such as random-walk Metropolis (RWM), MALA, and Hamiltonian Monte Carlo (HMC) fail to be geometrically ergodic \citep{yang2022stereographic, 
 jarner2000geometric,
 roberts1996geometric, livingstone2019geometric}, resulting in slow mixing of the corresponding Markov chains. \cref{fig:dark_side2} compares the trajectories of our SCS (with parameters chosen as in \cref{sec: tuning parameters}), SPS, together with other popular algorithms such as RWM and HMC. The left panel shows the traces of the squared norm started at stationary and the middle panel shows the first two coordinates when the chains are started at the tails. In the latter case, SPS is stuck at initialization, and we therefore omit its output. The empirical acceptance probability is approximately 0.234 for SCS, SPS and RWM, and 0.7 for HMC, as prescribed by the literature. 
 
 There is a simple geometric reason explaining the different behavior of SCS and SPS, with the latter corresponding to SCS when the observer $o$ is placed at the north pole. For the stereographic projection (SP), the transformed density on the surface of the sphere is unbounded near the north pole, because the tails of the multivariate Cauchy density are too heavy. In contrast, for any other choice of $o$ inside the sphere, the density induced by SCP remains bounded on the bright side; see \cref{fig:dark_side1} for an illustration in 2 dimensions. This seemingly subtle difference between the two projections has important consequences: when initialized in the tails, SPS tends to become trapped near the north pole (the dark side of the moon), whereas SCS escapes the dark side and is able to efficiently explore the target distribution.

\begin{figure}[ht!]
    \centering
    \includegraphics[width=0.32\linewidth]{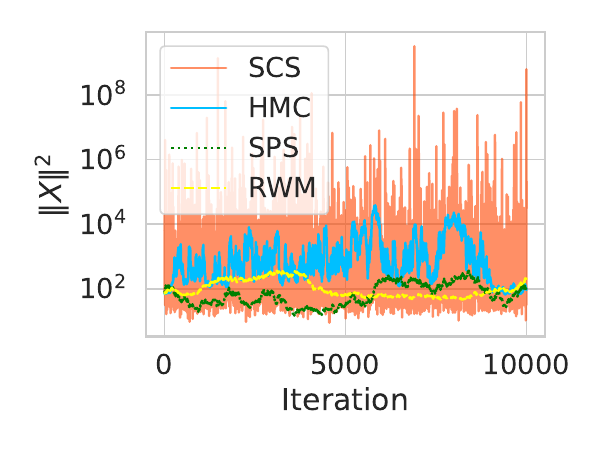}
    \includegraphics[width=0.32\linewidth]{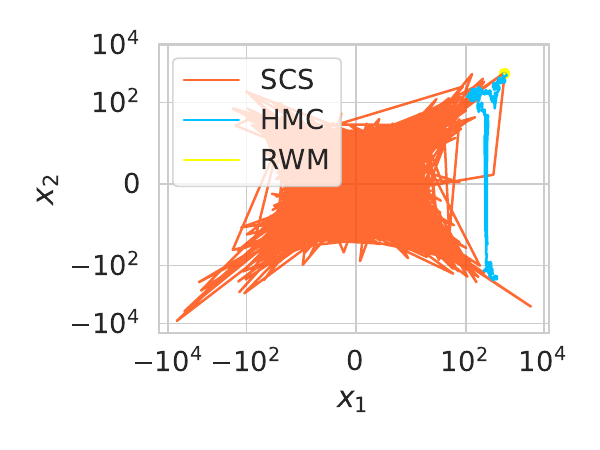}
    \includegraphics[width=0.32\linewidth]{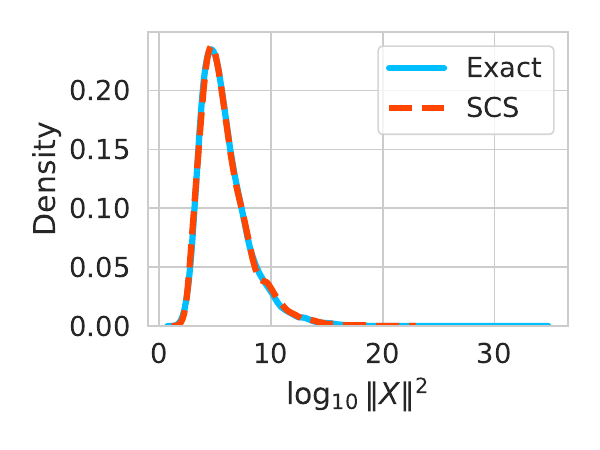}
    \caption{Sampling from a 100-dimensional standard Cauchy distribution. 
    Left panel: trace plot of $\|X\|^2$ of SCS (orange),  HMC (blue), SPS (green), and RWM (yellow), initializing from a sample drawn from the target.
    Middle panel: trajectory in the first two dimensions (on the log scale) of SCS (orange),  HMC (blue), and RWM (yellow), initializing from $(10^3, 10^3,\ldots)$. 
    Right panel: density of the SCS samples of $\|X\|^2$ (red, estimated by kernel density regression), compared against the true density (blue). 
    }
    \label{fig:dark_side2}
\end{figure}
\begin{figure}[ht!]
    \centering
    \includegraphics[width=0.8\linewidth]{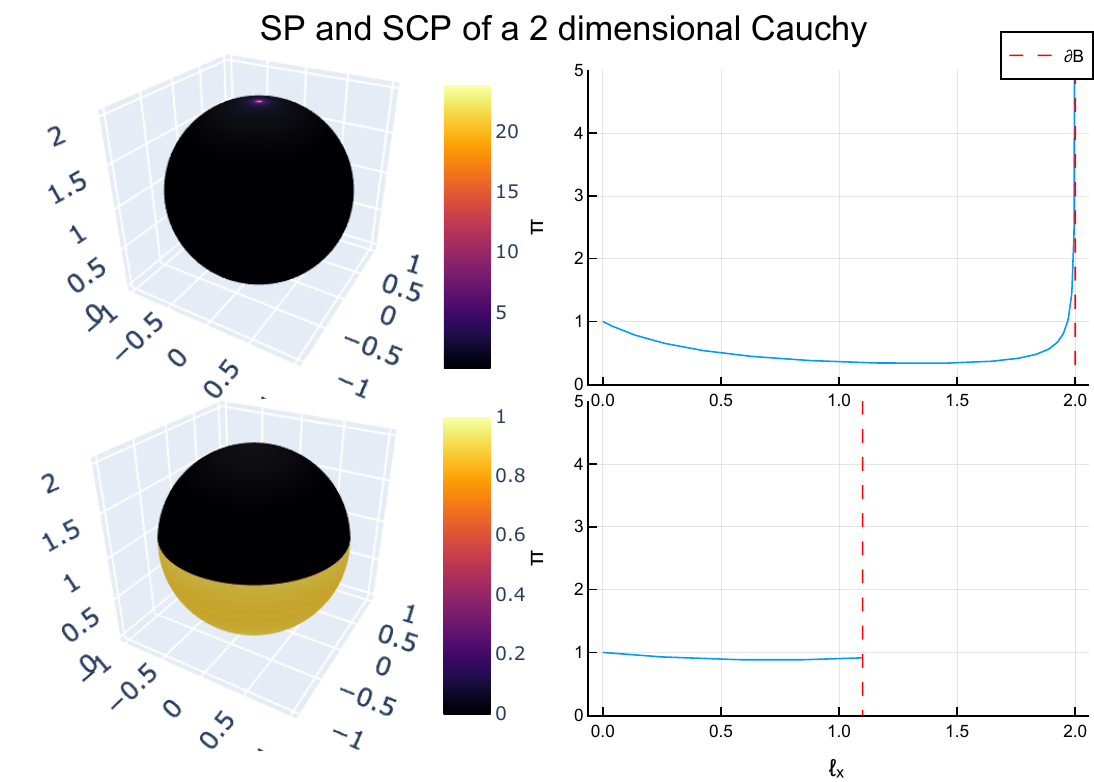}
    \caption{SP (top panels) and SCP with observer $o=(0,\, 0,\, 1.1)$ (bottom panels) of 2-dimensional Cauchy distribution. Right panels correspond to the unnormalized density evaluated along the latitude component $\ell_x$. The density on the sphere peaks to infinity on the north pole for SP, while it is bounded from above and below for SCP.
     }
    \label{fig:dark_side1}
\end{figure}

\subsection{Related work}
Our work overcomes current limitations of popular MCMC methods, which are not geometrically ergodic for heavy-tailed targets \cite[Sec.~4]{jarner2000geometric} \cite[Thm.~2.2]{livingstone2019geometric},\cite[Prop.~2]{wang2025stereographic}. Some transformation strategies have been proposed to regain geometric/uniform ergodicity \citep{johnson2012variable, yang2022stereographic}, with the former not recovering uniform ergodicity and the latter not suitable for $d$-dimensional distributions without finite $(d-1)$-th moment \citep[Thm 2.1 and Thm 2.2]{yang2022stereographic}\citep[Sec 3.1.3]{brevsar2025central}. { A notable exception is the Independence Sampler (IS) \citep[][]{mengersen1996rates}, which is uniformly ergodic for targets whose tails are at most as heavy as the proposal. 
However, it is well known that the IS convergence rate typically deteriorates (exponentially) quickly as the dimension increases for common choices of proposal, making this method often impractical for heterogeneous targets in high dimensions.}

The past years have seen a growing interest in sampling algorithms for heavy-tailed distributions. For example, \citep{he2022heavy, mousavi2023towards,  vasdekis2023speed, he2024mean, bertazzi2025sampling} derived theoretical guarantees for Langevin dynamics and piecewise deterministic Markov processes. However, existing analyses focus either on the continuous-time dynamics or the so-called unadjusted setting, which introduces a bias to the target that is difficult to quantify in practical scenarios. Our focus is on Metropolis-adjusted algorithms for heavy-tailed target distributions, for which mostly negative results have been derived \citep{roberts1996geometric,jarner2000geometric, livingstone2019geometric}, although
\cite{jarner2007convergence} showed that the use of heavy-tailed MCMC proposal distributions can mitigate slow polynomial rates of convergence to some extent.

\subsection{Outline and notation}
The remainder of the paper is organized as follows. \cref{sec: SC Projection} provides a detailed description of the Sub-Cauchy Projection. 
\cref{sec: SC sampler} introduces our proposed algorithm, Sub-Cauchy Projection Sampler (SCS). \cref{sec: theory} establishes the uniform ergodicity of the proposed sampler for sub-Cauchy distributions. \cref{sec: tuning parameters} describes how the tuning parameters of the proposed method are automatically calibrated. 
\cref{sec: numerics} numerically tests our algorithm for complex heavy-tailed high-dimensional distributions. Code for reproducing the numerical results is available at \url{https://github.com/liusf15/sub-cauchy-sampler}. More technical proofs are deferred to Appendix~\ref{sec:appendix}.

\textbf{Notation.} 
For a vector $z \in \RR^{d+1}$, let $h_z := z_{1:d} \in \RR^d$ and $\ell_z := z_{d+1} \in \RR$ be respectively the horizontal component (longitude) and the vertical component (latitude) of $z$. For \( d \ge 1 \) and \( s \in \RR^{d+1} \), let  
$\overline{\cB}^{d+1}(s) := \{ x \in \RR^{d+1} : \|x - s\| \le 1 \}$, $ 
\cS^{d}(s) := \{ x \in \RR^{d+1}: \|x - s\| = 1 \}$, 
and 
$\cB^{d+1}(s) := \overline{\cB}^{d+1}(s) \setminus  \cS^{d}(s).
$ 
For a fixed \( o \in \cB^{d+1}(s) \), the bright side of the sphere $\cS^{d}(s)$ is defined as
\begin{equation*}
    \cC^{d}_o(s) := \{x \in \cS^{d}(s) \colon x_{d+1} < o_{d+1} \}.
\end{equation*}
When $s = (0,0,\dots,0)$, we omit the argument and simply write $\cB^{d+1},\, \cS^{d}, \, \cC^{d}_o$. For all $1\leq i\leq d+1$, $e_i$ denotes $i$th standard basis vector in $\RR^{d+1}$.

\section{Sampling sub-Cauchy distributions}
\label{sec: Methodology}

\subsection{Sub-Cauchy Projection}\label{sec: SC Projection}
Let the observer be $o \in \cB^{d+1}(e_{d+1})$.  The $d$-dimensional Sub-Cauchy Projection  
$
\mathrm{SCP}_o : \cC^{d}_o(e_{d+1}) \to \RR^d
$
maps any point \( x \in \cC^{d}_o(e_{d+1}) \) to \( y \in \RR^d \) such that \( x, o \), and \( \bar y= (y, 0) \in \RR^{d+1} \) are colinear, that is,
\[
x = t\, o + (1 - t)\, \bar y, 
\quad \text{for some } t \in [0,1].
\]
By further shifting and rescaling \( y \) as \( yR + \mu \), where \( \mu \in \RR^d \) and \( R > 0 \), we generalize the projection to spheres of radius \( R \) centered at \( (\mu, R) \). 
We denote this projection as $\mathrm{SCP}_\theta$, where $\theta=(o,\mu,R)$.
Explicit expressions of the Sub-Cauchy Projection, its inverse, and its Jacobian are given below. The derivation is provided in Appendix~\ref{sec: forward_backward_jacobian}.
\begin{proposition} \label{prop:bijection_and_jacobian}
Fix $\theta = (o ,\,\mu, \, R)$, where $o \in \cB^{d+1}(e_{d+1}),\,\mu \in \RR^d$, and $R > 0$. 
The Sub-Cauchy Projection takes the form 
\begin{equation*}
y = \mathrm{SCP}_\theta(x) = R\left(\frac{\ell_o}{\ell_o-\ell_x}h_x - \frac{\ell_x}{\ell_o - \ell_x} h_o\right) + \mu, \qquad x \in \cC_o^{d}(e_{d+1}),
\end{equation*}
with inverse
\begin{equation*}
\begin{pmatrix}
    h_x\\ 
    \ell_x
\end{pmatrix}
     = \mathrm{SCP}^{-1}_\theta(y) = 
    \begin{pmatrix}
         M \hat y + (1-M) h_o\\
         (1-M)\ell_o
     \end{pmatrix}, \qquad \hat y = (y - \mu)/R, \qquad y \in \RR^d,
\end{equation*}
where
\begin{align}
    \label{eq: M}
    M = \frac{ -\big(\langle \hat y - h_o, h_o \rangle - \ell_o(\ell_o - 1)\big) + \sqrt{ \big(\langle \hat y - h_o, h_o \rangle - \ell_o(\ell_o - 1)\big)^2 - (\|\hat y - h_o\|^2 + \ell_o^2)(\|h_o\|^2 + \ell_o^2 - 2\ell_o) } }{ \|\hat y - h_o\|^2 + \ell_o^2}.
\end{align}
The Jacobian of the forward map $\mathrm{SCP}_\theta$ is 
\begin{align}
\label{eq: jacobian}
    J_\theta(y) 
     &= R^d\left(\frac{ M \|\hat y - h_o\|^2  + \langle \hat y - h_o, h_o\rangle  + \ell_o - \ell_o^2 (1-M)}{M^d\ell_o}\right).
\end{align}
\end{proposition}

The Jacobian greatly simplifies when the observer is centered in its longitudinal component, that is when $h_o = (0,0,\dots,0)\in\RR^d$, particularly in the two special cases below, where Example~\ref{example 2} is obtained by extending  $o \in \cB^{d+1}(e_{d+1})$ to also include the boundary point $(0,0,\dots,0,2)\in\RR^{d+1}$.
\begin{example}(Cauchy projection)\label{example 1}
If $h_o = (0,0,\dots,0)\in\RR^d$ and $\ell_o = 1$, then
\begin{equation}
    \label{eq:jacobian_subcauchy}
    J_\theta(y) = R^d( \|\hat y\|^2  + 1)^{(d+1)/2}, \quad \text{where } \hat y = (y -\mu)/R.
\end{equation}
In this case, $J_\theta(y)^{-1}$ is proportional to the density of a standard $d$-dimensional Cauchy distribution (when $\mu = (0,0,\dots,0),\, R = 1$). Therefore, the inverse of $\mathrm{SCP}_\theta$ maps the $d$-dimensional Cauchy distribution into the uniform distribution on the bright side of the moon.
\end{example} 
\begin{example}(Stereographic projection) \label{example 2}
 If $h_o = (0,0,  \dots, 0)\in\RR^d$ and $\ell_o = 2$, then we recover the classical stereographic projection with 
\begin{equation}
    \label{eq:jacobian_stereographic}
    J_\theta(y) = R^d\left(\frac{\|\hat y\|^2 + 4}{4}\right)^d, \quad \hat y = (y -\mu)/R.
\end{equation}
Here $J_\theta(y)^{-1}$ is proportional to the density of a standard multivariate student's $t$ with degrees of freedom $d$ (when $\mu = (0,0,\dots,0),\, R = \sqrt{d}/2$). 
\end{example}
\begin{remark}\label{rmk: jacobian}
   The inverse of Jacobian in \cref{eq:jacobian_subcauchy} decays slower than that in  \cref{eq:jacobian_stereographic} by a factor of $\|\hat y\|^{d-1}$ as $\|y\|\to\infty$. This allows us to improve the uniform ergodicity result in \cite[Theorem 2.1]{yang2022stereographic} for targets that have sub-Cauchy tails.
\end{remark}

\subsection{Sub-Cauchy Projection Sampler}\label{sec: SC sampler}
We detail here our proposed Sub-Cauchy Projection Sampler (SCS). SCS is a random-walk-type Metropolis algorithm on the bright side of the moon. The Metropolis acceptance--rejection ensures that the chain satisfies the detailed-balance condition and is $\pi$-invariant. Special care must be taken near the boundaries, as proposals landing on the ``dark side'' (i.e., above the observer $o$) are re-located on the ``bright side'', without breaking symmetries and detailed balance. For simplicity, we define the algorithm on unit spheres centered at 0.

For a given point $x \in \cS^{d}$, a new point $x'$ is first proposed by sampling a Gaussian random variable on the tangent space of the unit sphere at $x$ and projecting it back onto $\cS^{d}$; see Algorithm~\ref{alg:rwm proposal}. This initial proposal coincides with the SPS proposal in \citep[Algorithm 1]{yang2022stereographic} and allows us to leverage existing theoretical results. If the initial proposed point  $x'$ lands on the bright side, we let $x^{\star}=x'$. Otherwise, $x'$ lands in the dark side, i.e., $\ell_{x'} > \ell_o - 1$, we bring it back to the bright side with a \emph{stepping-out} function given by
\begin{equation}
    \label{eq: stepping-out}
    x^{\star}:=\cos(K\alpha)x + \sin(K\alpha)u,
\end{equation}
in which 
$
{K:=\min \left\{k\in\mathbb{N}: k\alpha> \phi+\gamma\right\}}$, $u:=(x'-(\langle x, x'\rangle )x)/\sqrt{1-(\langle x, x'\rangle )^2}$, \sloppy{$\alpha=\arccos(\langle x,  x'\rangle),$}
and where
$
\phi:=\arccos(x_{d+1}/\sqrt{x_{d+1}^2+u_{d+1}^2})$,  $\gamma:=\arccos((\ell_o-1)/\sqrt{x_{d+1}^2+u_{d+1}^2})$; see \cref{fig: stepping out} for an illustration.

Note that the stepping-out step does not incur additional computational cost, since \cref{eq: stepping-out} admits a closed-form expression. 
Throughout, we set $\ell_o \ge 1$, and we show in Appendix~\ref{app: latitude} that for suitable choices of $\ell_o$, the initial proposal $x'$ lies on the bright side with high probability as the dimension increases, so that \cref{eq: stepping-out} is invoked only sporadically. The proposed point is then accepted with probability 
\begin{align}\label{eq: alpha}
\alpha(y, y^\star) = \frac{J_\theta(y^\star) \pi(y^\star)}{J_\theta(y) \pi(y)} \wedge 1,
\end{align}
where $y, y^\star$ are respectively $x, x^\star$ projected in the Euclidean space via $\mathrm{SCP}_\theta$. If the proposed point is rejected, the new state is set to be equal to $x$; see Algorithm~\ref{alg:rwm} for the full description of the transition. 

As the proposal is symmetric on the sphere and the stepping-out function in \cref{eq: stepping-out} is symmetric relative to the proposal density, SCS satisfies the detailed-balance condition, and the following proposition holds. The proof is provided in Appendix~\ref{sec:invariance}.
\begin{proposition}\label{prop:invariance}
    If $\pi(y)$ is positive and continuous in $\mathbb{R}^d$, then SCS gives rise to an ergodic Markov chain in $\mathbb{R}^d$ with invariant distribution $\pi$.
\end{proposition}

\begin{algorithm}[!h]
        \caption{Random-walk proposal on unit sphere.\label{alg:rwm proposal}} 
        \KwIn{Current state $x\in\cS^{d}$, step-size $h > 0$.}
  Sample $\tilde{\delta}  \sim \cN(0, h^2 I_{d+1})$\;
     Set $\delta = \tilde \delta - \langle x, \tilde{\delta} \rangle x$\; 
\KwOut{$x' = \frac{x + \delta}{\|x + \delta\|}$.}
\end{algorithm}
\begin{algorithm}[!h]
        \caption{Sub-Cauchy Projection Sampler (SCS) transition (on unit sphere). \label{alg:rwm}} 
        \KwIn{Current state $x\in\cS^{d}$, tuning parameters $\theta = (o, R, \mu)$, step-size $h > 0$.}
  Propose $x' \in\cS^{d}$ with Algorithm~\ref{alg:rwm proposal}\;
 If $\ell_{x'} > \ell_o - 1$ ($x'$ is in dark side), set $x^\star$ as in \cref{eq: stepping-out}, otherwise set $x^\star = x'$\;
Compute $\alpha (y,y^\star)$ as in \cref{eq: alpha},  where $y = \mathrm{SCP}_\theta(x + e_{d+1})$, $y^\star = \mathrm{SCP}_\theta(x^\star + e_{d+1})$\;
With probability $\alpha(y,y^\star)$, set $x^{\mathrm{new}} = x^\star$; otherwise set $x^{\mathrm{new}} = x$\;
\KwOut{next state $x^{\mathrm{new}}$.}
\end{algorithm}

\begin{figure}[ht!]
    \centering
    \resizebox{0.3\textwidth}{!}{
    \input{files/tikz_illustration}
            }
    \caption{Top view of the stepping-out mechanism. The shaded area represents the dark side. Starting from $x$, the point $x'$ was proposed in the dark side and the stepping-out function brings it back to the bright side $x^\star$. \label{fig: stepping out}}
\end{figure}

\subsection{Uniform Ergodicity}\label{sec: theory}
Our main result shows that the convergence of the SCS chain to a sub-Cauchy target $\pi$ is exponential, and its speed does not depend on the starting value. Formally, a Markov chain is \emph{uniformly ergodic} if, for all $y \in \RR^d$, 
$$
\|\cP^n(y, \cdot) - \pi\| \le C \rho^n, 
$$
for some constants $C>0$, $0<\rho < 1$, where $\|\cdot\|$ denotes total variation distance and $\cP^n(y, \cdot)$ denotes the law of the Markov chain initialized at $y$ after $n$ steps. Uniform ergodicity is a strong  (yet qualitative) property that is rarely satisfied for Metropolis--Hastings chains in $\RR^d$. 

\begin{theorem}\label{thm_uniform_ergodicity}
    If $\pi(y)$ is positive and continuous on $\mathbb{R}^d$, then SCS with any fixed $\ell_o\in [1,2)$ is uniformly ergodic if and only if 
    \[
    \sup_{y\in\mathbb{R}^d}\pi(y)J_\theta(y)<\infty,
    \]
    where $J_\theta(y)$ is given in \cref{eq: jacobian}.
\end{theorem}

\cref{thm_uniform_ergodicity} implies that the proposed SCS is uniformly ergodic for a large class of heavy/light-tailed distributions, including distributions with tails no heavier than a (multivariate) Cauchy distribution; see \cref{cor: unif ergodicity}. To the best of our knowledge, SCS is the first general-purpose Metropolis-adjusted MCMC with such provable convergence guarantees.

\begin{corollary}[Uniform ergodicity of SCS for sub-Cauchy distributions] \label{cor: unif ergodicity}
      If $\pi(y)$ is sub-Cauchy, positive and continuous in $\mathbb{R}^d$, 
    then SCS with any fixed $\ell_o\in [1,2)$ is uniformly ergodic.
\end{corollary}
\begin{proof}
    By \cref{thm_uniform_ergodicity}, it suffices to show $ \sup_{y\in\mathbb{R}^d}\pi(y)J_\theta(y)<\infty$.
  Since
    \[
 \sup_{y\in\mathbb{R}^d}\pi(y)J_\theta(y)\le \left[\sup_{y\in \mathbb{R}^d}\pi(y)(1+\|y\|^2)^{(d+1)/2}\right]\cdot \left[\sup_{y\in \mathbb{R}^d}(1+\|y\|^2)^{-(d+1)/2}J_\theta(y)\right],
    \]
    and the first term is finite by Definition~\ref{def: sub-cauchy},
    it suffices to show
    \[
     \sup_{y\in \mathbb{R}^d}(1+\|y\|^2)^{-(d+1)/2}J_\theta(y)<\infty.
    \]
      As $R, \mu$, and $\ell_o<2$ are constants, this is verified by substituting $J_\theta(y)$ in \cref{eq: jacobian}, using the continuity of $\pi$, and taking the limit as $\|y\|\to\infty$.
\end{proof}
\begin{remark} \label{rmk: limitation SPS}
     \cref{cor: unif ergodicity} offers a substantial improvement to the uniform ergodicity result of SPS in \citet[Theorem 2.1]{yang2022stereographic}, which holds only for $d$-dimensional targets with lighter tails than a student's $t$ with degree of freedom $\nu \ge d$. In particular, SPS is not geometrically ergodic for student's $t$ with $\nu<d$ \cite[Proof of Theorem 2.1]{yang2022stereographic}\cite[Sec.~3.1.3]{brevsar2025central}; see also the numerical comparison in \cref{sec: dark side of the moon}.
\end{remark}

\section{Tuning parameters}
\label{sec: tuning parameters}

In this section, we provide a principled approach to select the tuning parameters of SCP before running the MCMC algorithm. For a given latitude of the observer $\ell_o \in [1,2)$, the goal is to choose the parameter $\bar \theta = (h_o, \mu, R)$ such that the target $\pi$ transformed to the bright side $\cC_o^{d}(e_{d+1})$ is close to the uniform distribution. This is because when the density on the bright side is uniform, Algorithm~\ref{alg:rwm} has acceptance probability equal to 1 for any step-size $h>0$, leading to fast mixing and small auto-correlation. 

Throughout, the latitude $\ell_o$ is fixed (see Appendix~\ref{app: latitude} for guidance motivated by simple geometric arguments), while the step-size $h>0$ of SCS is tuned in order to achieve a certain average empirical acceptance probability \citep{gelman1996efficient, roberts1997weak}, leaving the total number of iterations as the only user-specified parameter. 
\begin{remark}
    Popular adaptive schemes for tuning MCMC hyperparameters, \citep[e.g.,][]{gilks1998adaptive, bell2024adaptive} are based on estimating and matching the first two moments of the target distribution. This strategy may fail for heavy-tailed distributions, as these distributions do not necessarily have finite moments. 
\end{remark}

For a given $\ell_o \in [1,2)$, let $\bar \theta := (h_o, \mu, R) \in \Theta$ where $\Theta = \{h_o\in\RR^d: \|h_o\|_2^2 + (\ell_o-1)^2\leq 1 \} \times\RR^d \times \RR^{+}$. Let $q_{\bar \theta}$ denote the distribution of $y=\mathrm{SCP}_\theta(x)$, where $x\sim \mathrm{Unif}(\cC_o^{d}(e_{d+1}))$ and $\theta = (\ell_o, \bar \theta)$. 
We optimize $\bar \theta$ by minimizing the Kullback-Leibler (KL) divergence between $q_{\bar \theta}$ and $\pi$, that is,
$\min_{\bar \theta\in\Theta}\; \mathrm{KL}(q_{\bar \theta} \| \pi ).$ 
This approach falls within the framework of variational inference \citep{blei2017variational}, where one seeks to approximate the target $\pi$ by a distribution from the variational family $\{q_{\bar\theta}: \bar\theta\in\Theta \}$. 

The objective function can be expressed as 
\begin{equation}\label{equ: kl}
\begin{split}
    \mathrm{KL}(q_{\bar \theta} \| \pi ) &= \mathbb{E}_{q_{\bar \theta}}[ \log q_{\bar \theta}(y) - \log \pi(y) ]\\
    &=\mathbb{E}_{\mathrm{Unif}(\cC_o^{d}(e_{d+1})) }[\log q_{\bar \theta}(\mathrm{SCP}_\theta(x) ) - \log \pi(\mathrm{SCP}_\theta(x) ) ]\\
    &=\mathbb{E}_{\mathrm{Unif}(\cC_o^{d}(e_{d+1})) } [-\log J_\theta({\mathrm{SCP}_\theta(x)}) - \log \pi(\mathrm{SCP}_\theta(x) ) ] {+ \textrm{ const.}} 
\end{split}
\end{equation}
{ where the constant only depends on $\ell_o$.} The expectation in the objective function can be approximated using Monte Carlo samples drawn from $\mathrm{Unif}(\cC_o^{d}(e_{d+1}))$. The parameter $\bar\theta$ is then optimized via stochastic gradient descent, with gradients computed by automatic differentiation, following standard practice in black-box variational inference \citep{ranganath2014black}. 

If the target $\pi$ belongs to the variational family, then SCP with the optimal $\bar\theta$ exactly transforms $\pi$ to the uniform distribution on the bright side $\cC_o^{d}(e_{d+1})$. 
When $\pi$ does not lie in the variational family, the induced density on the bright side will not be exactly uniform. Nevertheless, the location parameter $\mu$ and scale parameter $R$ allow the variational approximation to capture the bulk of the target distribution, while the observer $o$ provides an additional degree of freedom to accommodate skewness. 
{ The scale parameter $R$ can also be generalized to allow for different scales for different dimensions \citep[Section 4]{yang2022stereographic}. In this case, $R \in \RR^{d\times d}$ is taken to be an invertible affine transformation of SCP. The transformation can be learned through the variational approach described above by parameterizing \(R\) via its Cholesky factor. This provides additional flexibility to adapt the projection for anisotropic and correlated targets, such as those considered in Section~\ref{sec:robit}. This dense parameterization comes at an additional cost: the total number of hyperparameters scales quadratically with the dimension, making it less suitable for very high-dimensional problems.} We provide an illustration of the performance of the variational approximation in \cref{sec:skew-cauchy}.

\section{Numerical simulations}
\label{sec: numerics}
\subsection{Multivariate skew \texorpdfstring{$t$}{t} distribution}\label{sec:skew-cauchy}
We consider here the multivariate skew $t$ distribution introduced by \citep{branco2001general,azzalini2003distributions}. Following \cite{azzalini2003distributions}, a $d$-dimensional skew $t$ distribution with location $\xi\in\RR^d$, scale matrix $\Omega\in\RR^{d\times d}$, skewness $\alpha\in\RR^d$, and degree of freedom $\nu$ can be represented as the distribution of $\xi+V^{-1/2} Z$, where $V\sim \chi^2_\nu / \nu$, and $Z$ follows a skew normal distribution with density
\begin{align*}
    2\, \varphi_d(z;0, \Omega) \,\Phi(\alpha^\top \omega^{-1} z ), 
\end{align*}
with $\omega=\mathrm{diag}(\Omega_{11}^{1/2},\ldots,\Omega_{dd}^{1/2} )$ and where $\Phi$ denotes the standard Normal cumulative distribution function and $\varphi_d(y ; \, \mu, \Sigma)$ denotes the $d$-dimensional Gaussian density with mean $\mu$ and variance $\Sigma$ evaluated at $y$. 
We fix $d=100$, $\xi=\mathbf{0}$, $\Omega=I_d$, and $\alpha=(100,-100,0,\ldots,0)$. We compare SCS (with latitude $\ell_o=1.1$, {scale parameter $R  \in \RR^+$,} and step-size $h=0.1$ on the sphere) against the standard HMC sampler, which takes 10 leapfrog steps of size 0.1 in each iteration. Both samplers are initialized at $(1,\ldots,1)$, run for 500,000 iterations, discard the first 100 samples, and retain every 50th sample for evaluation. 
The average acceptance rates are about 0.4 for SCS and 0.8 for HMC. 
We compute empirical marginal quantiles at probabilities \sloppy{$0.01,0.02,0.05,0.1, 0.2, 0.5, 0.8, 0.9, 0.95, 0.98, 0.99$}, and compare them against the corresponding exact quantiles, estimated from $10^7$ independent samples directly drawn from the target distribution.

\cref{fig: skewt qq} displays Q--Q plots for the first four coordinates, with variability across 10 independent replicates indicated by the shaded regions.
When $\nu=2$ (top panels), both samplers capture the bulk of the distribution, though HMC is less accurate at the extreme tails. When $\nu=1$ (bottom panels), corresponding to the skew Cauchy distribution, HMC fails to explore the distribution, whereas SCS continues to perform well, with accurate quantiles up to the 1st and 99th percentiles.
Notably, SCS is gradient-free, while HMC uses 10 leapfrog steps per iteration and therefore requires repeated gradient evaluations.

\begin{figure}
    \centering
    \includegraphics[width=\textwidth]{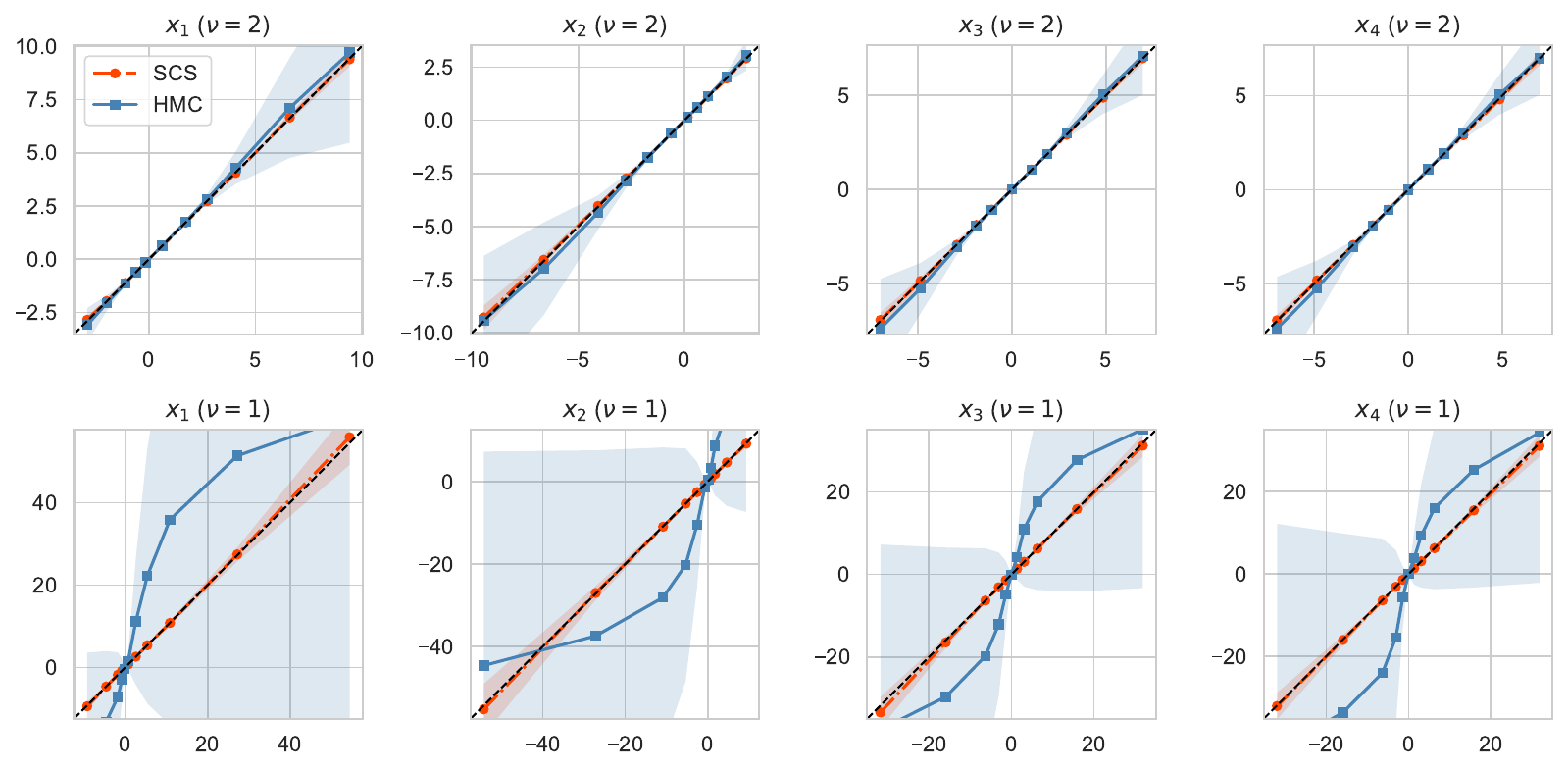}
    \caption{Q--Q plot of SCS (red) and HMC (blue) in sampling a 100-dimensional multivariate skew $t$ distribution with 2 (top) or 1 (bottom) degrees of freedom. Shaded regions show variation across 10 independent replicates.}
    \label{fig: skewt qq}
\end{figure}

{ In addition, we evaluate performance in estimating tail probabilities $\mathbb{P}(\|X\| > c)$ for $\nu=1$ and thresholds \(c\in\{2^4,2^5,\ldots,2^{15}\}\).
The left panel of \cref{fig:tail_probability} plots the estimated tail probabilities against the threshold $c$, together with reference values estimated from $10^7$ exact samples. The error bars represent the standard deviation across 10 independent chains. The right panel reports the relative standard deviation, defined as the standard deviation divided by the reference probability.  
The figure shows that SCS provides more stable estimates compared to HMC over a wide range of tail thresholds. 
The robustness to tail estimation can be informally justified by the uniform ergodicity property, which suggests that the mixing properties of the Markov chain are uniform across the state space. This is in contrast to non-uniformly ergodic Markov chains, which may explore the bulk of the target reasonably well while still exhibiting poor mixing in remote regions of the state space.

\begin{figure}
    \centering
    \includegraphics[width=0.7\linewidth]{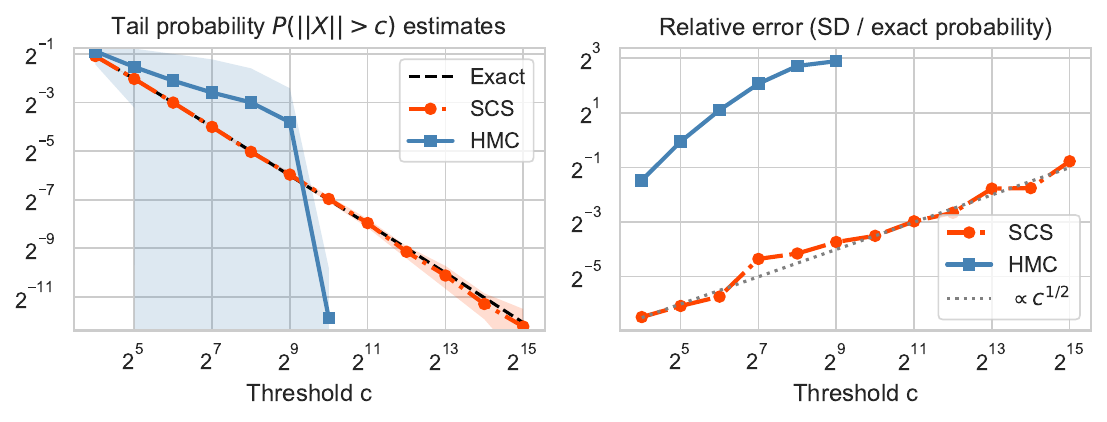}
    \caption{ {Monte Carlo estimate of $\PP(\|X\| > c)$ across 10 independent chains, where $X$ is a multivariate skew-Cauchy random variable and the corresponding relative errors (right panel) for different values of $c$.The dotted line is proportional to $c^{1/2}$, which is the theoretical scale of the relative error of i.i.d. samples.}}
    \label{fig:tail_probability}
\end{figure}
}

To illustrate the performance of the variational method described in \cref{sec: tuning parameters}, we report the cosine between $h_o$ and $\alpha$, as well as the relative difference between $\mu$ and $\xi$. In the target distribution, $\alpha$ and $\xi$ respectively control skewness and location. Their counterparts in the variational family are $-h_o$ and $\mu$, which play analogous roles. Hence, we expect $h_o$ and $\alpha$ to become aligned and $\mu$ to get close to $\xi$ when using the variational method of \cref{sec: tuning parameters}. 

In this experiment, we fix $d=100,\nu=2$. We use the same $\alpha$ as before, and set the elements of $\xi$ to be equally spaced between $-10$ and $10$. Optimization of the KL divergence \cref{equ: kl} is carried out with a Monte Carlo sample size of 2000 and the Adam optimizer \citep{kingma2014adam} with a learning rate of 0.01.
As shown in \cref{fig: skewt_vi_convergence}, $\cos(h_o, \alpha)$ approaches $-1$ as the number of gradient descent steps increases, while $\|\mu-\xi\|_2/\|\xi\|_2$ approaches 0.

\begin{figure}[ht!]
    \centering
    \includegraphics[width=0.7\textwidth]{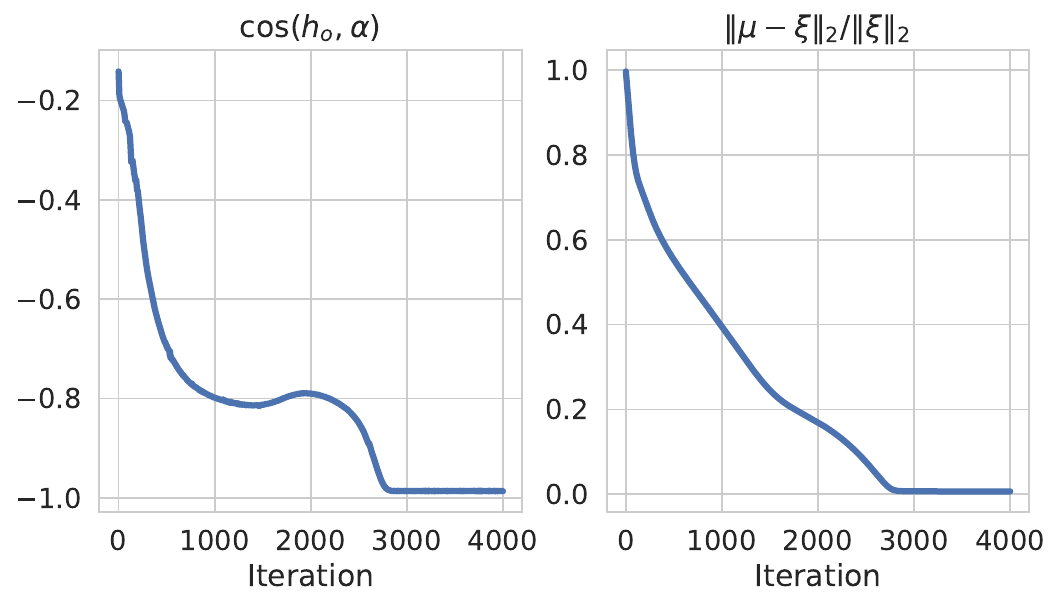}
    \caption{Cosine between $h_o$ and $\alpha$ (left) and relative distance between $\mu$ and $\xi$ (right) versus the number of gradient descent iterations when optimizing the KL divergence~\eqref{equ: kl}.}
    \label{fig: skewt_vi_convergence}
\end{figure}

To further examine the transient behavior, \cref{fig: skewt trace} shows the trace plots of a single chain from each sampler, both initialized from a cold start at $(500,\ldots,500)$, with $d=10$ and $\nu=2$. HMC fails to reach the typical set and continues to wander in the heavy-tailed region even after 100,000 iterations, whereas SCS rapidly converges to the typical set. This illustrates the substantial advantage of SCS over HMC when the chain is initialized far in the tail.

\begin{figure}[ht!]
    \centering
    \includegraphics[width=0.7\textwidth]{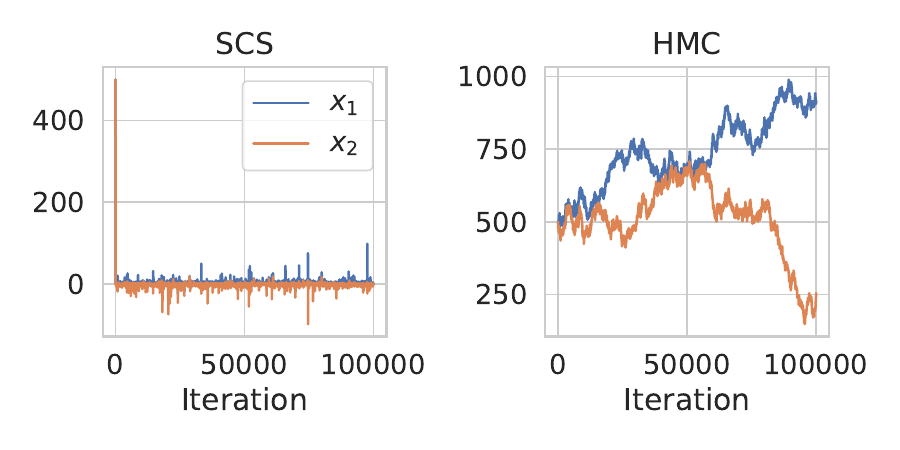}
    \caption{Trace plots of the first two coordinates of SCS (left) and HMC (right) in sampling a 10-dimensional multivariate skew $t$ distribution with 2 degrees of freedom.}
    \label{fig: skewt trace}
\end{figure}

\subsection{Robust Bayesian binary regression}\label{sec:robit}
We consider here  Bayesian binary regression problems, where the binary response $y_i\in\{0,1\}$ is modeled as 
\begin{align*}
    \mathbb{P}(y_i=1\mid x_i,\beta) = F(x_i^\top \beta),\quad 1\leq i\leq n,
\end{align*}
with covariates $x_i\in\RR^d$ and regression coefficients $\beta\in\RR^d$. This formulation reduces to the logistic regression model when $F(x)=\frac{1}{1+e^{-x}} $, and to the probit model when $F=\Phi$ is the CDF of the standard normal distribution. Both models are known to be sensitive to outliers, motivating the use of heavier-tailed link functions such as the CDF of a student's $t$ distribution \citep{albert1993bayesian,liu2004robit}, which is known as robit regression. 
When $\beta$ has a normal prior, geometric ergodicity has been established under suitable conditions for a Gibbs sampler with data augmentation (DA-Gibbs) \citep{roy2012convergence,mukherjee2023convergence}. 

\citet{gelman2008weakly} propose to use independent $t$ distributions as the default prior for the regression coefficients in logistic regression and other generalized linear models (e.g., multinomial logit and probit). Such priors encode weak prior information and, importantly, stabilize the estimates when the responses are linearly separable. More recently, \citet{ghosh2018use} establish necessary and sufficient conditions for the existence of posterior means under Cauchy priors, and develop a DA-Gibbs for logistic regressions with Cauchy priors. 

However, \citet{ghosh2018use} report extremely slow mixing of DA-Gibbs with heavy-tailed priors in the presence of separation. While they acknowledge that alternative samplers such as the No-U-Turn Sampler (NUTS; \citep{hoffman2014no}) can improve mixing, they note that ``the problematic issue of posteriors with extremely heavy tails under Cauchy priors cannot be resolved without altering the prior.'' { In Appendix A.5, we describe DA-Gibbs under data separation and provide a
heuristic argument for slow mixing of the chain. In particular, we show that,
as $n\to\infty$, the inverse square root of the negative log-posterior
curvature at the mode in the $\beta$-parameterization is of order
$\Theta(\sqrt{n\log n})$, whereas the DA-Gibbs step-size near the mode is of
order $\Theta(\sqrt{\log n})$. Thus, the step-size is smaller by a factor of
order $\sqrt n$, providing a local curvature-based heuristic explanation for
high autocorrelation and slow mixing.}

{We evaluate the proposed SCS method on three Bayesian binary regression experiments: logistic regression with student's \(t\) priors, robit regression with student's \(t\) priors, and logistic regression with a horseshoe prior \citep{carvalho2009handling}. In all three experiments, we use a synthetic separable data set with \(d=20\) covariates and \(n=50\) observations. The covariates are generated independently from a standard normal distribution, standardized to have mean zero and standard deviation \(0.5\), and the responses are set to
$y_i=\mathbf{1}\{x_{i,1}>0\}$,
so that the data are separable along the first covariate.

In the logistic-\(t\) experiment, we use the logistic link and place independent student's \(t_2\) priors on the regression coefficients, with scale \(2.5\), following \citet{gelman2008weakly}. In the robit-\(t\) experiment, we replace the logistic link by a student's \(t\) link with 2 degrees of freedom, while keeping the same student's \(t_2\) prior on the coefficients. In the logistic-horseshoe experiment, we use the logistic link and place a horseshoe prior on the regression coefficients to induce sparsity.

We compare four samplers: the data-augmentation Gibbs sampler, HMC (10 leapfrog steps per iteration), independence sampler (IS) with a Cauchy proposal, and the proposed SCS. Although the independence sampler with a Cauchy proposal can yield a uniformly ergodic Markov chain in these settings, its empirical performance is poor because of the heterogeneity of the posterior. For SCS, we use a dense scale matrix \(R\in\mathbb{R}^{d\times d}\). We set the latitude to our default value \(\ell_o=1.1\) for the student's \(t\) prior experiments. For the horseshoe prior experiment, we instead use \(\ell_o=1.7\)  which we found to improve the convergence of the initial variational approximation used to tune $\mu$ and $R$.

Each method is repeated independently over 10 chains. Each chain uses 100 warm-up iterations. SCS, the independence sampler, and Gibbs are run for 500,000 iterations, while HMC is run for 100,000 iterations, keeping the wall-clock times roughly comparable. After thinning, each chain retains 10,000 samples for quantile estimation. Reference quantiles are obtained using NUTS with 20 parallel chains and much larger computational resources.

\Cref{fig:logistic-t,fig:robit-t,fig:logistic-horseshoe} report Q--Q plots for the marginal posterior quantiles of \(\beta_1\) in the three experiments. The shaded bands indicate one standard deviation across the 10 independent chains. Across these examples, SCS gives stable quantile estimates in heavy-tailed settings, while Gibbs, HMC, and the independence sampler can show substantial discrepancies in the tails.
}

\begin{figure}
    \centering
    \includegraphics[width=\textwidth]{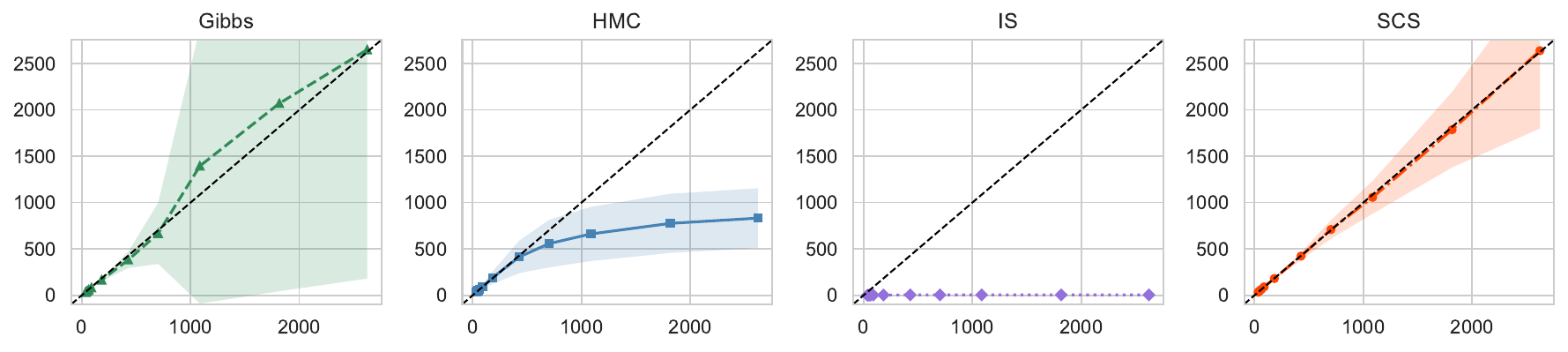}
    \caption{Q--Q plots for the marginal posterior quantiles of \(\beta_1\) in the logistic regression experiment with student's \(t\) priors. Columns correspond to Gibbs, HMC, the IS, and SCS. The horizontal axis gives the NUTS reference quantiles, and the vertical axis gives the corresponding empirical quantiles from each method. Shaded bands indicate one standard deviation across 10 independent chains.}
    \label{fig:logistic-t}
\end{figure}

\begin{figure}
    \centering
        \includegraphics[width=\textwidth]{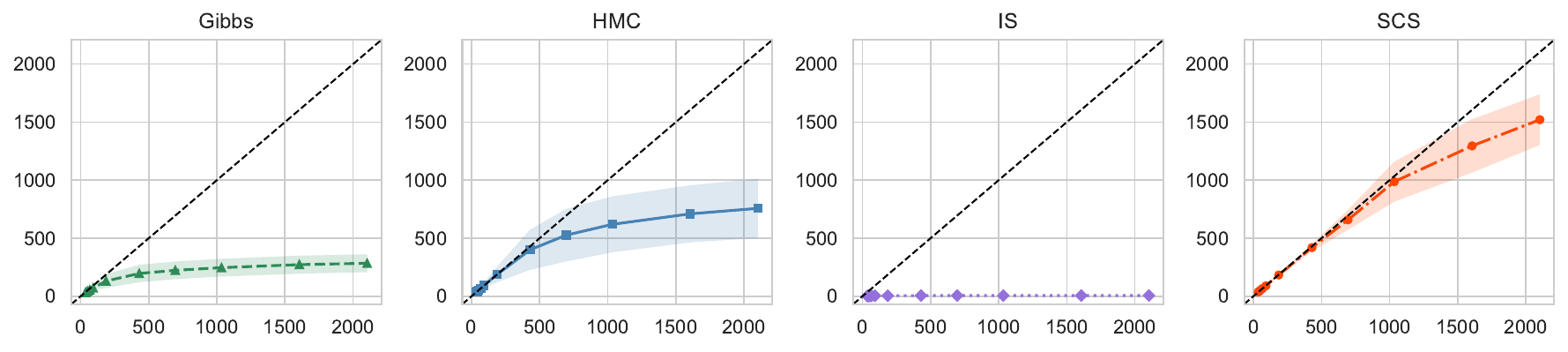}
    \caption{Q--Q plots for the marginal posterior quantiles of \(\beta_1\) in the robit regression experiment with student's \(t\) priors. The plotting convention is the same as in \cref{fig:logistic-t}.}
    \label{fig:robit-t}
\end{figure}

\begin{figure}
    \centering
    \includegraphics[width=\textwidth]{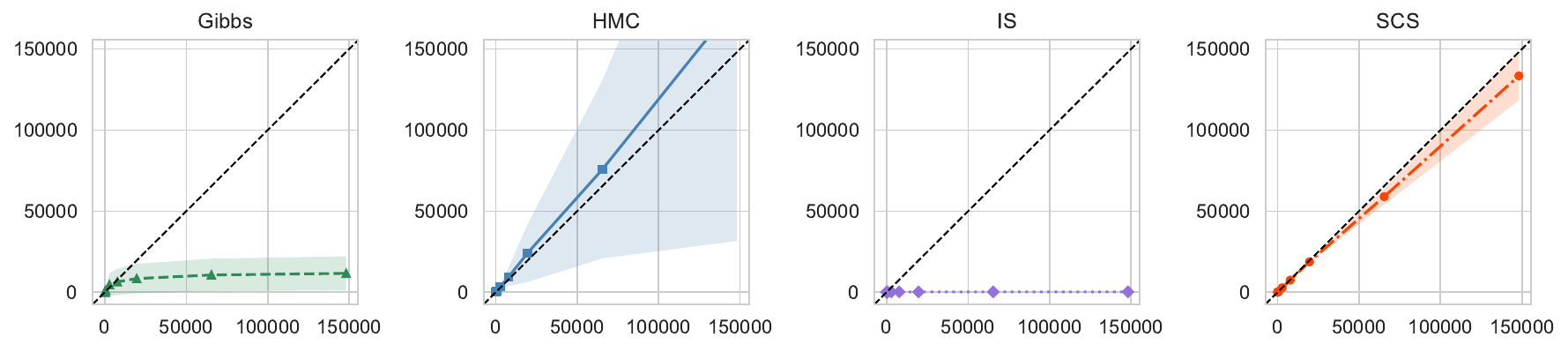}
    \caption{Q--Q plots for the marginal posterior quantiles of \(\beta_1\) in the logistic regression experiment with a horseshoe prior. The plotting convention is the same as in \cref{fig:logistic-t}.}
    \label{fig:logistic-horseshoe}
\end{figure}

\section{Discussion}
We introduce a new Metropolis--Hastings algorithm, Sub-Cauchy Projection Sampler, that is uniformly ergodic for sub-Cauchy target distributions. The tuning parameters of the projection are automatically set via the variational method, making the proposed SCS broadly applicable whenever the posterior is suspected to be heavy-tailed. Our numerical experiments show that SCS outperforms several popular methods in these scenarios. We conclude the paper by outlining several open problems and directions for future research.
\begin{enumerate}
    \item \emph{Optimal scaling and adaptive MCMC}: we proposed a variational approach for selecting tuning parameters that works well across our simulations. This approach does not tune the latitude $\ell_0$, whose adaptive selection is left for future work. Furthermore, the optimality of this tuning rule remains an open question. A popular framework to adaptively tune hyperparameters in Metropolis--Hastings algorithms is given by 
    the optimal scaling theory \citep{roberts1997weak}. However, in our framework with heavy-tailed targets and geometrically transformed MCMC algorithms, the derivation of scaling limits is challenging \citep{yang2020optimal,yang2022stereographic}. We therefore defer a rigorous treatment of this issue to future work.

    \item \emph{Quantitative convergence results}: Quantitative bounds on the mixing time of MCMC algorithms and their dependence on the dimension have been an active area of theoretical research in recent years \citep{dwivedi2019log,chen2020fast,yang2023complexity,andrieu2024explicit,milinanni2026ongoing}. An important direction for future work is to establish rapid mixing guarantees for SCS over a broad class of sub-Cauchy target distributions.

    \item \emph{Gradient-based MCMC for sub-Cauchy distributions}: SCS is a random-walk–type algorithm that does not exploit gradient information of the target distribution. {  However, SCP is a geometric transformation that is not intrinsically tied to a random-walk proposal. An important direction for future work is to integrate SCP with ideas from robust gradient-based MCMC methods \citep{livingstone2022barker,bell2026ongoing}, or to combine it with MCMC frameworks directly defined on the sphere such as Geodesic Slice Sampler \citep{habeck2025geodesic}, as well as 
    non-local preconditioning \citep{girolami2011riemann}, or boundary conditions \citep{bierkens2023sticky, bierkens2023methodsapplicationspdmpsamplers}.}
\end{enumerate}

\section*{Acknowledgement}
The authors would like to thank Nicolas Chopin for helpful discussions on Bayesian robust regression and Wenkai Xu for helpful discussions at the early stage of this work. The authors also would like to thank Miha Bresar, Daniel Rudolf and Cameron Bell for their feedback on the first version of this manuscript. SG acknowledges support from the European Union (ERC), through the Starting Grant ‘PrSc-HDBayLe’, project number 101076564, and from the Danish Data Science Academy (DDSA), through the visit grant number 2025-5427. JY’s research is supported by the Independent Research Fund Denmark (DFF) through the Sapere Aude Starting Grant (5251-00032B).
GOR was supported by PINCODE (EP/X028119/1), EP/V009478/1 and by the UKRI grant, OCEAN, EP/Y014650/1.

\bibliographystyle{biometrika}
\bibliography{files/reference}

\pagebreak
\appendix
\section{Appendix}\label{sec:appendix}
\subsection{Sub-Cauchy Projection and Jacobian determinant}\label{sec: forward_backward_jacobian}

For a fixed observer $o = (h_o, \ell_o) \in \cB^{d+1}(e_{d+1}) \subset \RR^{d+1}$, a point on the surface of the unit sphere $x = (h_x, \ell_x) \in \cC^{d}(e_{d+1}) \subset \RR^{d+1}$ and rescaling factors $R>0$, $\mu \in \RR^d$, let $\hat y = (y-\mu)/R$ be the rescaled point on the plane. The forward map 
\begin{equation*}
\label{eq: linear map}
y =\mathrm{SCP}_\theta(x) =  R\left(\frac{\ell_o}{\ell_o-\ell_x}h_x - \frac{\ell_x}{\ell_o - \ell_x} h_o\right) + \mu, \qquad h_x = x_{1_d}, \, \ell_x = x_{d+1}, 
\end{equation*}
is simply derived as the straight line connecting $o$, $x$ and $(\hat y, 0)$. To find the inverse map, we use that \[
\begin{pmatrix}
           h_{x}\\
           \ell_x -1
         \end{pmatrix}
         = M
         \begin{pmatrix}
           \hat y\\
           -1
         \end{pmatrix} + (1-M)
        \begin{pmatrix}
           h_o \\
           \ell_o -1
         \end{pmatrix},
         \]
for some $M \in [0,1]$. Since
$\|h_x\|^2 + (\ell_x-1)^2 = 1$, 
by taking the squared norm on both sides of the equation above, we have that
\[
1 = M^2 \left(\|\hat y - h_o\|^2 + \ell_o^2\right) +  \|h_o\|^2 + (\ell_o - 1)^2 + 2M\left( \langle \hat y - h_o, h_o\rangle - \ell_o (\ell_o -1)\right)
\]
which has one positive solution $ M$ that can be written explicitly in terms of $y, o, \ell_o$ and $R$ (see \eqref{eq: M}). Hence the \emph{inverse map} takes the form
\begin{equation*}
\label{eq: inverse map 1}
\begin{pmatrix}
    h_x\\ 
    \ell_x
\end{pmatrix}
     = [\mathrm{SCP}^{-1}_\theta(y)] = 
    \begin{pmatrix}
         M \hat y + (1-M) h_o\\
         (1-M)\ell_o
     \end{pmatrix}.
\end{equation*}

Let us derive the Jacobian determinant for the map $y= \pi_o(x)$. Let $\bar y = (\hat y, 0)$, $\theta_1$ be the angle between $\bar y - o$ and $x - e_{d+1}$, $\theta_2$ be the angle between $o - \bar y$ and $e_{d+1}$, $a = \|x-o\|$ and $b = \|\bar y - x\|$. Then, the Jacobian determinant is defined as
\begin{align*}
    J(y) =R^d\left( \frac{a+b}{a}\right)^{d}\frac{cos(\theta_1)}{\cos(\theta_2)}  &= \frac{R^d}{M^d}  \frac{\langle  \bar y - h_o, x - e_{d+1}\rangle}{\langle o - \bar y, e_{d+1}\rangle}\\
    &= R^d\left(\frac{\langle \hat y - h_o, M \hat y + (1-M) h_o\rangle  + \ell_o - \ell_o^2 (1-M)}{M^d\ell_o}\right)\\
     &= R^d\left(\frac{ M \|\hat y - h_o\|^2  + \langle \hat y - h_o, h_o\rangle  + \ell_o - \ell_o^2 (1-M)}{M^d\ell_o}\right).
\end{align*}

For $h_o = 0_d, \,\ell_o = 2$ (\cref{example 2}), we have
\[
M = \frac{4}{\|\hat y\|^2 + 4}
\]
and we recover the stereographic projection with Jacobian
\[
J_{\theta}(y) = R^d \frac{1}{M^d} = R^d \left(\frac{\|\hat y\|^2 + 4}{4}\right)^d.
\]
For $h_o =  0_d, \,\ell_o = 1$ (\cref{example 1}), we obtain the Cauchy projection with
\begin{equation*}
    M = \frac{1}{\sqrt{\|y\|^2 + 1}} 
    \label{eq: M cauchy}
\end{equation*}
and Jacobian
\begin{equation*}
J_{\theta}(y) = R^d \, \frac{1}{M^d} = R^d \frac{\|\hat y\|^2 + 1}{M^{d-1}} =  R^d( \|\hat y\|^2  + 1)^{(d+1)/2}.
\end{equation*}

\subsection{Proof of \texorpdfstring{\cref{prop:invariance}}{invariance}}\label{sec:invariance}
We begin with the following lemma, {in which $\pi$ denotes the angle $\pi$ radians.}
\begin{lemma}\label{lem:SCP}
Given two points $x$ and $y$ on unit sphere $\cS^{d}\subset \RR^{d+1}$, there is a unique greatest circle passing both $x$ and $y$, which can be written as
\[
c_o(\theta)=\cos(\theta)x + \sin(\theta)u,\quad \theta\in [0,2\pi),
\]
where $u= \frac{y-(x^Ty)x}{\sqrt{1-(x^Ty)^2}}$.
Suppose the angle between $x$ and $y$ are no larger than $\pi/2$  with $x_{d+1}<\ell_o-1$ and $y_{d+1}>\ell_o-1$. The two intersection points of this greatest circle with latitude $\ell_o-1$ are:
\[
p_{-}:=x\cos(\phi-\gamma)+u \sin(\phi-\gamma)
\]
which lies on the shorter arc on that greatest circle connecting $x$ and $y$ (aka, \emph{spherical linear interpolation} or the \emph{slerp}) and
\[
p_{+}:=x\cos(\phi+\gamma) + u \sin(\phi+\gamma),
\]
which lies on the longer complementary arc, and
\[
\phi:=\arccos\left(\frac{x_{d+1}}{\sqrt{x_{d+1}^2+u_{d+1}^2}}\right), \quad \gamma:=\arccos\left(\frac{\ell_o-1}{\sqrt{x_{d+1}^2+u_{d+1}^2}}\right).
\]
\end{lemma}
\begin{proof}
    For $c_o(\theta)$, one can verify $c_o(0)=x$ and $c_o(\alpha)=x(x^Ty)+y-(x^Ty)x=y$, where $\alpha:=\arccos(x^Ty)<\pi/2$ by our assumption. Also, $u$ is a unit vector, as
    \[
    u=\frac{y-(x^Ty)x}{\sqrt{1-(x^Ty)^2}}=\frac{y-(x^Ty)x}{\|y-(x^Ty)x\|}.
    \]
    The two intersection points should satisfy
    \[
    c_o(\theta)_{d+1}=x_{d+1}\cos\theta + u_{d+1}\sin\theta = \ell_o-1,
    \]
    which gives the solutions $\theta=\phi\pm \gamma$. Finally, the one with $\phi-\gamma$ is on the shorter arc.
\end{proof}

By \cref{lem:SCP}, one can easily check that the SCS proposal is symmetric (see \cref{fig: stepping out} for an illustration) and therefore the detailed balance condition holds for SCS, implying the existence of the stationary distribution $\pi$. By \citep[Theorem 13.0.1]{Meyn2012}, \cref{prop:invariance} holds if the Markov chain is aperiodic, {$\pi$-irreducible}, and positive Harris recurrent. SCS is aperiodic by \cite[Theorem 3(i)]{roberts1994simple} using the facts that the acceptance probability does not equal zero and the proposal distribution as a Markov kernel is aperiodic. Furthermore, $\pi$-irreducibility comes from \cite[Theorem 3(ii)]{roberts1994simple} using the fact that the proposal distribution as a Markov kernel is $\pi$-irreducible. Since $\pi$ is finite, the chain is positive recurrent \cite[pp.1712]{tierney1994markov}. Harris recurrence comes directly from \cite[Corollary 2]{tierney1994markov}.

\subsection{Proof of \texorpdfstring{\cref{thm_uniform_ergodicity}}{uniform ergodicity}}
    First of all, if $\sup_{y\in\mathbb{R}^d}\pi(y)J(y)=\infty$, SCS is not even geometrically ergodic \cite[Proposition 5.1]{roberts1996geometric}. This implies that the condition in \cref{thm_uniform_ergodicity} is necessary. Therefore, it suffices to show $\sup_{y\in\mathbb{R}^d}\pi(y)J(y)<\infty$ is also a sufficient condition. 
    
Recall that $\cS^{d}$ denotes the unit sphere centered at the origin, that is, $\cS^{d}:=\{z\in \Reals^{d+1}: \sum_{i=1}^{d+1} z_i^2=1\}$, and let $\pi_S(z)$ be the transformed target on the unit sphere. Note that when $\ell_o<2$, part of the sphere is not used, which we refer to as the ``dark side'' of the sphere, defined by $\DS(0)$ in the following. We will also define ``almost dark side'' with size $\delta$, denoted as $\DS(\delta)$. When $\ell_o=2$, our definition of $\DS(\delta)$ reduces to the definition of ``Arctic Circle'', denoted as $\AC(\delta)$, in \cite[Proof of Theorem 2.1]{yang2022stereographic}. 
	\begin{enumerate}
		\item We define the ``dark side'' of the sphere as $$ \DS(0):=\{(x,\ell_x-1)\in \cS^{d}: \ell_x\ge \ell_o\}.$$
        \item Similarly, we can define the ``almost dark side'' with ``size'' $\delta\ge 0$ as
     $$\DS(\delta):=\{(x,\ell_x-1)\in \cS^{d}: \ell_x\ge \ell_o-\delta\},$$
     which contains a small area of the boundary of the ``bright side'' with ``size'' $\delta$.

		\item Same as in \cite[Proof of Theorem 2.1]{yang2022stereographic}, we use $\HS(z,0):=\{z'\in\cS^{d}: z^Tz'\ge 0\}$ to denote the ``hemisphere'' from $z$. and $\HS(z,\epsilon'):=\{z'\in\cS^{d}: z^Tz'\ge \epsilon'\}$ to denote the area left after ``cutting'' the hemisphere's boundary of ``size'' $\epsilon'\ge 0$. 
	\end{enumerate}

     The rest of the proof is a generalisation of  \cite[Proof of Theorem 2.1]{yang2022stereographic}, which is a special case for $\ell_o=2$, to any given $\ell_o\in [1,2]$. We prove uniform ergodicity by showing the compactification of the ``bright side'' of the sphere, denoted as $\overline{\text{DS}^c(0)}$, is a small set. More precisely, we will show the $3$-step transition kernel satisfies the following minorisation condition:
	\[
	P^3(z,A)\ge \epsilon \pi_S(A), \quad z\in \overline{\DS^c(0)}
	\]
	for all measurable set $A\subset \overline{\DS^c(0)}$, where $\epsilon>0$. {The rationale behind establishing a $3$-step minorisation condition is that this facilitates 
 the Markov chain on the sphere being able to reach any point on the ``bright side'' of the sphere, 
  and there exists enough probability mass around ``$2$-stop'' paths to ensure that the chain can avoid stopping too close to the boundary of the ``dark side'', the potentially problematic region.}

	Now we almost completely follow \cite[Proof of Theorem 2.1]{yang2022stereographic}. Consider the $3$-step path $z\to z_1\to z_2\to z'\in A$ where $z\in\overline{\DS^c(0)}$ is the starting point, $z_1$ and $z_2$ are two intermediate points on the sphere, and $z'\in A$ is final point. 
		Denoting $q(z,\cdot)$ as the proposal density, we have
	\[
	p(z,z')\ge q(z,z')\left(1\wedge \frac{\pi_S(z')}{\pi_S(z)}\right).
	\]
	Then the  $3$-step transition kernel can be bounded below by
\begin{equation}
    \begin{split}
	P^3(z,A)
	&\ge\int_{z'\in A}\int_{z_2\in\cS^d}\int_{z_1\in\cS^d} q(z,z_1)q(z_1,z_2)q(z_2,z')\\
	&\qquad \cdot\left(1\wedge\frac{\pi_S(z_1)}{\pi_S(z)}\right)\left(1\wedge\frac{\pi_S(z_2)}{\pi_S(z_1)}\right)\left(1\wedge\frac{\pi_S(z')}{\pi_S(z_2)}\right)\nu(\dee z_1)\nu(\dee z_2)\nu(\dee z') 
    \end{split}
\end{equation}
	where $\nu(\cdot)$ denotes the Lebesgue measure on $\cS^{d}$. 

	By the condition we assumed, $\pi(y)>0$ is continuous in $\Reals^d$, which implies $\pi_S(z)$ is positive and continuous in any compact subset of the ``bright side'' of the sphere $\DS^c(0)$. Then, using the fact that a continuous function on a compact set is bounded, if we rule out $\DS(\epsilon)$, $\pi_S(z)$ is bounded away from $0$. That is,
	\[
	\inf_{z\notin \DS(\epsilon)}\pi_S(z)\ge \delta_{\epsilon}>0, \forall \epsilon\in (0,1),
	\]
	where $\delta_{\epsilon}\to 0$ as $\epsilon\to 0$.
	
	Furthermore, since the surface area of $\cS^{d}$ is finite and $\pi_S(z)\propto \pi(y)J(y)$, we know
	$\sup_{y\in \Reals^d}\pi(y)J(y)<\infty$ implies that $M:=\sup \pi_S(z)<\infty$. Therefore, we have
\begin{equation}
    \begin{split}
        	P^3(z,A)
	&\ge\int_{z'\in A}\int_{z_2\in\cS^d\setminus \DS(\epsilon)}\int_{z_1\in\cS^d\setminus\DS(\epsilon)} q(z,z_1)q(z_1,z_2)q(z_2,z')\\
	&\qquad \cdot\left(1\wedge\frac{\delta_{\epsilon}}{M}\right)^2\left(\frac{\pi_S(z')}{M}\right)\nu(\dee z_1)\nu(\dee z_2)\nu(\dee z').
    \end{split}
\end{equation}
	Next, we consider the term 
	$q(z,z_1)q(z_1,z_2)q(z_2,z')$.
	For our algorithm, the proposal can cover the whole hemisphere except the boundary; by ``cutting'' the boundary a little bit, the proposal density will be bounded below. That is, we have
	\[
	q(z,z')\ge \delta'_{\epsilon'}>0, \quad \forall z'\in \HS(z,\epsilon'), \epsilon'\in (0,1),
	\]
	where $\delta'_{\epsilon'}\to 0$ as $\epsilon'\to 0$.
	
	Therefore, if 
	$z_1\in \HS(z,\epsilon')$ and $z_2\in \HS(z',\epsilon')$
	then 
	$q(z,z_1)\ge \delta'_{\epsilon'}$ and $q(z_2,z')\ge \delta'_{\epsilon'}$.
	Then we have
\begin{equation}
    \begin{split}
        	&\int_{z_2\in\HS(z',\epsilon')\setminus \DS(\epsilon)}\int_{z_1\in\HS(z,\epsilon')\setminus\DS(\epsilon)} q(z,z_1)q(z_1,z_2)q(z_2,z')\nu(\dee z_1)\nu(\dee z_2)\\
	&\ge (\delta'_{\epsilon'})^3\int_{z_2\in\HS(z',\epsilon')\setminus \DS(\epsilon)}\int_{z_1\in\HS(z,\epsilon')\setminus\DS(\epsilon)}\Ind_{z_1\in \HS(z_2,\epsilon')}\nu(\dee z_1)\nu(\dee z_2).
    \end{split}
\end{equation}

	Note that as the ``bright side'' is no smaller than a hemisphere, $3$ steps are enough to reach any point $z'$ from any $z$ both on the ``bright side''. Furthermore, as the ``dark side'' of the sphere is no larger than a hemisphere, it is clear that under the Lebesgue measure $\nu$ on $\cS^{d}$, we can define the following positive constant $C$:
	\[
	C:=\inf_{z,z'\in \DS^c(0)}\int_{z_2\in\HS(z',0)\setminus \DS(0)}\int_{z_1\in\HS(z,0)\setminus\DS(0)}\Ind_{z_1\in \HS(z_2,0)}\nu(\dee z_1)\nu(\dee z_2)>0.
	\]
	Now consider the following function of $\epsilon\ge 0$ and $\epsilon'\ge 0$
	\[
	C_{\epsilon,\epsilon'}:=\inf_{z,z'\in \DS^c(0)}\int_{z_2\in\HS(z',\epsilon')\setminus \DS(\epsilon)}\int_{z_1\in\HS(z,\epsilon')\setminus\DS(\epsilon)}\Ind_{z_1\in \HS(z_2,\epsilon')}\nu(\dee z_1)\nu(\dee z_2).
	\]
    It is clear that $C\ge 	C_{\epsilon,\epsilon'}$. { The difference between them can be bounded by
\begin{equation}
    \begin{split}
    C-	C_{\epsilon,\epsilon'}\le &\nu(\cS^d)\left[\sup_{z\in \DS^c(0)} \nu(\{z_1: z_1\in (\HS(z,0)\setminus\HS(z,\epsilon'))\cup (\DS(\epsilon)\setminus\DS(0))\})\right.\\
    &+\sup_{z'\in \DS^c(0)}\nu(\{z_2: z_2\in (\HS(z',0)\setminus\HS(z',\epsilon'))\cup (\DS(\epsilon)\setminus\DS(0))\})\\
    &\left.+\sup_{z_2\in \DS^c(0)}\nu(\{z_1: z_1\in \HS(z_2,0)\setminus\HS(z_2,\epsilon')\})\right].
    \end{split}
\end{equation}
}

Then clearly each term in the upper bound is continuous w.r.t.~both $\epsilon$ and $\epsilon'$ and goes to zero when $\epsilon\to 0$ and $\epsilon'\to 0$. Therefore, there exists $\epsilon>0$ and $\epsilon'>0$ such that $C-C_{\epsilon,\epsilon'}\le \frac{1}{2} C$, that is
	\[
	C_{\epsilon,\epsilon'}\ge\frac{1}{2}C>0.
	\]
	Using such $\epsilon$ and $\epsilon'$ we have
	\[
	\int_{z_2\in\HS(z',\epsilon')\setminus \DS(\epsilon)}\int_{z_1\in\HS(z,\epsilon')\setminus\DS(\epsilon)} q(z,z_1)q(z_1,z_2)q(z_2,z')\nu(\dee z_1)\nu(\dee z_2)\ge \frac{1}{2}C(\delta'_{\epsilon'})^3>0.
	\]
	Then we have
	\[
	P^3(z,A)\ge \left(1\wedge \frac{\delta_{\epsilon}}{M}\right)^2\frac{1}{2}C(\delta'_{\epsilon'})^3\int_{z'\in A} \left( \frac{\pi_S(z')}{M}\right)\nu(\dee z')
	\]
	Note that $\pi_S(z')>0$ almost surely w.r.t.~$\nu$ if $z'\in \DS^c(0)$. 
	We define
	\[
	g(z'):=\begin{cases}
	    \left(1\wedge \frac{\delta_{\epsilon}}{M}\right)^2\frac{1}{2}C(\delta'_{\epsilon'})^3 \frac{\pi_S(z')}{M}, &\quad \forall z'\in \DS^c(0)\\
        0,  &\quad \forall z'\in \DS(0).
	\end{cases}
	\]
	Note that $g(z)\propto \pi_S(z)$, Hence, we have established
	\[
	P^3(z,A)\ge \int_{z'\in A}g(z')\nu(\dee z')\ge \epsilon \pi_S(A),
	\]
	where $\epsilon=\left(1\wedge \frac{\delta_{\epsilon}}{M}\right)^2\frac{1}{2}C(\delta'_{\epsilon'})^3 \frac{1}{M}>0$. This establishes the desired minorisation condition and completes the proof.

\subsection{Choice for the latitude \texorpdfstring{$\ell_o$}{ell-o}}\label{app: latitude}
The variational method presented in Section~\ref{sec: tuning parameters} applies for a fixed latitude $\ell_o \in [1,2)$, as the support of the density on the bright side depends on this parameter. Our choice of $\ell_o$ is guided by simple geometric properties of the hyper-sphere in high dimensions. 
For $o \in \cB^{d+1}(e_{d+1})$, let $$\cD^{d}_o = \cS^{d}(e_{d+1}) \setminus \cC^{d}_o =  \{ x \in \RR^{d+1} \colon \|h_x\|^2 + (\ell_x - 1)^2 = 1; \ell_x \ge \ell_o\}$$ be the dark side of the sphere surface, that is removed in the SCP. The following result shows that the removed surface region has asymptotically negligible surface area compared to the sphere surface, when $\ell_o - 1 \gtrsim \frac{\log d}{\sqrt{d}}$.

\begin{proposition}\label{prop: ratio of volumes} For any $d \ge 1$ and $\ell_o \in (1,2)$, 
\begin{equation}
\label{eq: ratio}
    \frac{|\cD_o^{d}|}{|\cS^{d}|}
    \le {\frac{e^{1/2}}{4}}\sqrt{d+1}\left[1-(\ell_o-1)^2\right]^{\frac{d-1}{2}},
\end{equation}
where $|\cdot|$ is the $d$-dimensional surface area.
\end{proposition}
\begin{proof}
    Let \(X\) and \(Y\) be independent chi-square random variables with \(d\) and \(1\) degrees of freedom, respectively. Using the fact that the surface area ratio can be written as an angle ratio, we are interested in the probability
$$
    \frac{\textrm{surface area of ``dark side''}}{\textrm{surface of the sphere}}={\frac{1}{2}}\mathbb{P}\left(X \le \frac{1-(\ell_o-1)^2}{(\ell_o-1)^2} Y\right).
$$
We first start with the exact formula of
$  \mathbb{P}(X\le c Y)$,
for some constant $c$. Using the \(F\)-distribution, we have $\frac{X/d}{Y/1} \sim F_{d,1}$ which implies \[
  \mathbb{P}(X < cY)
=   \mathbb{P}\left(\tfrac{X}{Y} < c\right)
=   \mathbb{P}\left(d\,F_{d,1} < c\right)
=   \mathbb{P}\left(F_{d,1} < \tfrac{c}{d}\right)
= F_F\left(\tfrac{c}{d};\,d,1\right),
\]
where \(F_F(x;d,1)\) denotes the c.d.f.~function of an \(F\)-distribution with \((d,1)\) degrees of freedom. The c.d.f.~of the \(F\)-distribution can be written in terms of the regularized incomplete beta function:
\[
F_F(x;d,1)
= I_{\!\frac{d\,x}{d\,x+1}}\left(\tfrac{d}{2}, \tfrac{1}{2}\right),
\]
where
$I_z(a,b)
= \frac{1}{B(a,b)} \int_{0}^{z} t^{a-1} (1-t)^{b-1} \, dt$.
Substituting \(x = \tfrac{c}{d}\) yields the exact formula:
\[
  \mathbb{P}(X < cY)
= I_{\!\frac{d\,(c/d)}{d\,(c/d)+1}}\left(\tfrac{d}{2}, \tfrac{1}{2}\right)
= I_{\!\frac{c}{c+1}}\left(\tfrac{d}{2}, \tfrac{1}{2}\right).
\]
Next, we derive an upper bound for large enough dimension $d$. By Markov's inequality, for any \(t > 0\) such that the moment generating functions exist,
\[
  \mathbb{P}(X < cY)
=   \mathbb{P}\left(e^{-tX} \ge e^{-t\,cY}\right)
\le \mathbb{E}\left[e^{-tX + t\,cY}\right]
= M_X(-t)\,M_Y(tc).
\]
For \(Z \sim \chi^2_k\), the moment generating function is
\[
M_Z(u) = (1 - 2u)^{-k/2}, \quad u < \tfrac12.
\]
Hence, for \(0 < t < 1/(2c)\),
\[
  \mathbb{P}(X < cY)
\le (1 + 2t)^{-d/2}\,(1 - 2c\,t)^{-1/2}.
\]
Choosing $t^* = \frac{d - c}{2c\,(d + 1)}$ { when $c<d$},
one obtains the tightest bound of that form
\[
  \mathbb{P}(X < cY)
\le \left(1+\frac{1}{c}\frac{d-c}{d+1}\right)^{-d/2}\cdot\left(1-\frac{d-c}{d+1}\right)^{-1/2}= \sqrt{\frac{d+1}{1+c}}\left[\frac{d+1}{d}\cdot\frac{c}{c+1}\right]^{d/2}.
\]
Using $(1+1/d)^d<e$, we have
\[
  \mathbb{P}(X < cY)\le \sqrt{\frac{d+1}{c+1}}e^{1/2}\left(\frac{c}{c+1}\right)^{d/2}=\sqrt{d+1}e^{1/2}\sqrt{\frac{1}{c+1}\left(1-\frac{1}{c+1}\right)}\left(\frac{c}{c+1}\right)^{\frac{d-1}{2}}.
\]
Since $\sqrt{\frac{1}{c+1}\left(1-\frac{1}{c+1}\right)}\le \frac{1}{2}$, we have
\[
  \mathbb{P}(X < cY)\le \frac{e^{1/2}}{2}\sqrt{d+1}\left(\frac{c}{c+1}\right)^{\frac{d-1}{2}}.
\]
{ When $c\ge d$, the upper bound is trivial as it is larger than $1$:
\[
\frac{\sqrt{e}}{2}\sqrt{d+1}
\left(\frac{c}{c+1}\right)^{\frac{d-1}{2}}
\ge
\frac{\sqrt{e}}{2}\sqrt{d+1}
\left(1-\frac{1}{d+1}\right)^{\frac{d-1}{2}}
\ge
\frac{\sqrt{e}}{2}\sqrt{d+1}
\sqrt{\frac{2}{d+1}}
=
\sqrt{\frac{e}{2}}
>1.
\]
} 
Finally, substituting $c=\frac{1-(\ell_o-1)^2}{(\ell_o-1)^2}$ gives the desired result.
\end{proof}
\cref{prop: ratio of volumes} implies that, by restricting $\ell_o \in (1 + \frac{c}{\sqrt{d}}, 2)$, with $c \sim \log d $ the ratio \eqref{eq: ratio} shrinks fast as the dimension increases, making increasingly unlikely for the SPS proposal to land on the dark side. 
However, for Cauchy distributions, \cref{example 1} shows that the optimal choice corresponds to $\ell_o = 1$. In practice, we found $\ell_o \in [1.1,1.7]$ to be a reasonable choice in all our simulations, and we leave for future work the fine tuning of this parameter.

\subsection{Slow mixing of DA-Gibbs}
\label{sec: da gibbs}
For Bayesian logistic regression, it has been noticed that DA-Gibbs also performs poorly when the two classes are imbalanced. \citet{johndrow2019mcmc} studied the mixing behavior of DA-Gibbs in the infinitely imbalanced asymptotic regime introduced by \citet{owen2007infinitely}. \citet{johndrow2019mcmc} argue that the slow mixing is due to the discrepancy between the width of the posterior and the step-size of Gibbs (the step-size converges to 0 at a faster rate than the concentration of the posterior as the sample size $n\to\infty$). We notice a similar phenomenon when the two classes are perfectly separable, providing an explanation for the poor performance of DA-Gibbs in the empirical results of \cite{ghosh2018use}.

Consider a one-dimensional logistic regression with a student's $t$ prior with $\nu$ degrees of freedom
\[
\mathbb{P}(Y_i = 1 \mid \beta)
= \sigma(x_i \beta),
\qquad
\sigma(t) = \frac{1}{1+e^{-t}}
\]
and assume the following separable data
\[
(x_i,\, y_i) = 
\begin{cases}
     (n^{-1/2}, 1) & i \le n/2 \\
      (-n^{-1/2}, 0) & i>n/2 
\end{cases}
\]
for $i = 1,2,\dots, n$.
In the Gibbs sampler, the student's $t$ prior $t_\nu(0, \tau^2)$ is represented as a scale-mixture of Gaussians:
\[
    \beta\mid\lambda \sim N \big(0,\frac{\tau^2}{\lambda} \big),
    \quad
    \lambda \sim \mathrm{Gamma}\big(\frac{\nu}{2},\frac{\nu}{2}\big),
\]
and P\'olya-Gamma latent variables are used for the logistic link~\citep{polson2013bayesian}.

Define the rescaled coefficient $\eta =\beta/\sqrt n $.
In the $\eta$-parameterization, the log likelihood is
\[
\log p(y \mid \eta)
= \frac{n}{2} \log \sigma(\eta)
 + \frac{n}{2} \log (1-\sigma(-\eta))
= n \log \sigma(\eta),
\]
and the derivative w.r.t.\,$\eta$ is
\[
\frac{\partial}{\partial \eta} \log p(y\mid\eta)
= n \sigma(-\eta).
\]
The student's $t$ prior in $\beta$ induces the prior in $\eta$:
\[
\pi_\eta(\eta)
\propto 
\left(1 + \frac{n\eta^2}{\nu \tau^2}\right)^{-(\nu+1)/2}.
\]
Moreover, we have
\[
\frac{\partial}{\partial \eta}\log \pi_\eta(\eta)
= -\frac{(\nu+1)n\eta}{\nu\tau^2 + n\eta^2}.
\]
Combining the likelihood and prior, the derivative of the log posterior is given by
\[
\frac{\partial}{\partial \eta} \log p(\eta\mid y)
= n \sigma(-\eta)
  -\frac{(\nu+1)n\eta}{\nu\tau^2 + n\eta^2}.
\]

{ We now compare the inverse square root of the negative log-posterior
curvature at the mode with the scale of the Gibbs step-size in $\beta$
for large $n$.}
\begin{itemize}
\item \emph{Inverse curvature at the posterior mode.} We first argue that, for large $n$, the inverse square root of the curvature at the posterior mode is of order $\Theta(\sqrt{\log n})$.
For large $\eta$, $\sigma(-\eta)\sim e^{-\eta}$ and
$\frac{\partial}{\partial\eta}\log\pi_\eta(\eta)
\sim -\frac{\nu+1}{\eta}$,
so the mode $\hat\eta_n$ satisfies
\[
n e^{-\hat\eta_n}
\sim
\frac{\nu+1}{\hat\eta_n}
\quad\Longrightarrow\quad
\hat\eta_n \sim \log n.
\]
Differentiating again and evaluating at the mode gives
\[
-\left.
\frac{\partial^2}{\partial\eta^2}
\log p(\eta\mid y)
\right|_{\eta=\hat\eta_n}
=
n\sigma(\hat\eta_n)\sigma(-\hat\eta_n)
+
\frac{(\nu+1)n
\bigl(\nu\tau^2-n\hat\eta_n^2\bigr)}
{\bigl(\nu\tau^2+n\hat\eta_n^2\bigr)^2}
\sim
\frac{\nu+1}{\hat\eta_n}
=
\Theta\left(\frac{1}{\log n}\right).
\]
Therefore, the inverse square root of the negative log-posterior
curvature is
$\left\{
-\left.
\frac{\partial^2}{\partial\eta^2}
\log p(\eta\mid y)
\right|_{\eta=\hat\eta_n}
\right\}^{-1/2}
=
\Theta(\sqrt{\log n})$.
Since $\beta=\sqrt n\,\eta$, the corresponding inverse-curvature
quantity in the $\beta$-parameterization is
$\Theta(\sqrt{n\log n})$.

\item \emph{Scale of step-size. }
In the P\'olya-Gamma Gibbs sampler~\citep{polson2013bayesian}, the P\'olya-Gamma variables $\omega_i$ have the conditional distribution
\[
\omega_i \mid \eta \sim \mathrm{PG}(1, x_i\beta)
= \mathrm{PG}(1,\pm \eta),
\]
with mean given by
\[
\mathbb{E}[\omega_i\mid\eta]
= \frac{1}{2|\eta|}
  \tanh\!\left(\frac{|\eta|}{2}\right)
\sim \frac{1}{2|\eta|}
\quad(|\eta|\to\infty).
\]
Moreover, the conditional variance of $\beta$ given the P\'olya-Gamma variables is given by
\[
\Var{\beta\mid \omega, \lambda, y}
=
\Big( X^\top \Omega X + \tau^{-2}\lambda \Big)^{-1}.
\]
Note that
$
X^\top \Omega X
= \sum_i \omega_i x_i^2
= \frac{1}{n}\sum_i \omega_i.
$
Near the mode of $\eta$, we have $\mathbb{E}[\omega_i\mid \eta]\sim \frac{1}{2\log n } $. Thus $X^\top\Omega X\sim\frac{1}{2\log n}$.
Because
$
\lambda \mid \beta \sim 
\mathrm{Gamma}\!\left(\frac{\nu+1}{2},
                     \frac{\nu + \beta^2/\tau^2}{2}\right),
$
we have
$
\mathbb{E}[\lambda\mid\beta]
= \frac{\nu+1}{\nu + \beta^2/\tau^2}.
$
Near the posterior mode, $\beta^2 \sim n(\log n)^2$, thus
$\lambda/\tau^2
=O_{\mathbb P}\!\left(\frac{1}{n(\log n)^2}\right)$ is negligible compared to $X^\top\Omega X$. So we have
\[
\operatorname{Var}(\beta\mid\omega,\lambda,y)
\sim (X^\top\Omega X)^{-1}
\sim 2\log n,
\]
which means the step-size in $\beta$ is of order $\Theta(\sqrt{\log n})$.
\end{itemize}

In summary, the inverse square root of the negative log-posterior
curvature in the $\beta$-parameterization is
$\Theta(\sqrt{n\log n})$ at the mode, whereas the conditional
Pólya--Gamma Gibbs step-size is $\Theta(\sqrt{\log n})$. The latter is
smaller by a factor of order $\sqrt n$. This curvature comparison
provides a heuristic explanation for the slow mixing of the Gibbs
sampler.

\end{document}

%% file: files/tikz_illustration.tex
\begin{tikzpicture}[
  line width=1.2pt,
  draw=red!70!black,
  every node/.style={red!70!black
  },
]


\draw (0,0) circle (4.8);

\begin{scope}
  \clip (0,0) circle (2.7);
  \fill[
  pattern=north east lines, 
  pattern color=black!40!white] (-3,-3) rectangle (3,3);
\end{scope}

\draw (-2.5,2.0) circle (0.38);
\node at (-2.5,2.0) {$x$};


\draw[-{Latex[length=3mm]}]
  (-2.05, 1.9) .. controls (-1.0,1.9) and (0.4,2.3) .. (1.05,2.85);

\draw (1.40,3.1) circle (0.38);
\node at (1.40,3.1) {$x^{\star}$};


\node[rotate=07] at (-0.95,2.0) {$|$};
\node at (-0.90, 1.6) {$x'$};
\node[rotate=16] at (0.25,2.33) {$|$};

\draw[decorate, decoration={brace, amplitude=6pt, mirror, raise=3pt},]
  (-0.6,1.6) -- (1.4, 2.5)
  node[midway, below=13pt, rotate=25] {$k=2$};
\end{tikzpicture}